\documentclass[lettersize,journal]{IEEEtran}
\usepackage{amsmath,amssymb,amsfonts}
\usepackage{array}
\usepackage[caption=false,font=normalsize,labelfont=sf,textfont=sf]{subfig}
\usepackage{textcomp}
\usepackage{stfloats}
\usepackage{url}
\usepackage{verbatim}
\usepackage{graphicx}
\usepackage{cite}
\usepackage{graphicx,algorithm,algpseudocode,algorithmicx}
\usepackage{textcomp}
\usepackage{xcolor,cuted}
\usepackage{booktabs}
\usepackage{epstopdf,epsfig}
\usepackage{comment}
\usepackage[flushleft]{threeparttable}
\usepackage[font=small,labelfont=bf]{caption}
\usepackage{subfig}
\usepackage{tikz}
\usetikzlibrary{arrows,shapes,chains}
\newtheorem{lemma}{Lemma}
\newtheorem{theorem}{Theorem}
\newtheorem{corollary}{Corollary}

\newtheorem{remark}{Remark}
\newtheorem{definition}{Definition}
\newtheorem{proposition}{Proposition}
\newtheorem{example}{Example}
\newenvironment{proof}{\textbf{Proof:}}{\hfill$\square$}
\newcommand{\ie}{\textit{i.e.}}

\begin{document}

\title{State Feedback Control Design for Input-output Decoupling of Boolean Control Networks}

\author{Yiliang Li, Hongli Lyu, Jun-e Feng, and Abdelhamid Tayebi, \IEEEmembership{Fellow, IEEE}
\thanks{Corresponding authors: Jun-e Feng and Abdelhamid Tayebi.}
\thanks{Yiliang Li and Jun-e Feng are with School of Mathematics, Shandong University, Jinan, Shandong 250100, P.R. China (e-mail: liyiliang1994@126.com,fengjune@sdu.edu.cn).}
\thanks{Hongli Lyu and Abdelhamid Tayebi are with Department of Electrical Engineering, Lakehead University, Thunder Bay, Ontario, Canada (e-mail: \{hlyu1,atayebi\}@lakeheadu.ca).}
\thanks{This work was supported by the National Natural Science Foundation of China under the grant 62273201, the Research Fund for the Taishan Scholar Project of Shandong Province of China under the grant tstp20221103, and the National Sciences and Engineering Research Council of Canada (NSERC), under the grant NSERC-DG RGPIN 2020-06270.}}



\maketitle

\begin{abstract}
A state feedback control strategy is proposed for input-output (IO) decoupling of a class of fully output controllable Boolean control networks (BCNs). Some necessary and sufficient conditions for BCN IO-decoupling are presented. As an instrumental tool in our design, we introduce a canonical form for IO-decoupled BCNs along with some conditions guaranteeing its existence.
Finally, two numerical examples are provided to illustrate the effectiveness of the proposed approach.
\end{abstract}


\section{Introduction}\label{sec1}
In 1969, Kauffman suggested a general theory of metabolic behaviour, and studied the behaviour of large nets of randomly interconnected binary (on-off) genes \cite{SAKauffman1969}. These type of models referred to as Boolean networks (BNs), have been instrumental in the description and simulation of the behaviour of a variety of biological systems and physiological processes, such as genetic regulatory networks \cite{TIdenker2001}, cell growth, apoptosis and differentiation \cite{Huangsui2000}.
This type of discrete (sequential) dynamical systems, involving the interconnection of Boolean state variables, have attracted the attention of the control community due to their inherent challenges that required the development of appropriate tools to deal with this type of systems. The dazzling development of Boolean Control Networks (BCNs) is mainly attributed to the powerful algebraic framework developed by Cheng and co-authors \cite{Chengdaizhan2001}, that enabled the representation of these logical networks as linear state space representations with canonical state vectors. This allowed to recast most of the common systems and control problems into the BCN world (see, for instance, \cite{Zhuqunxi2018a,Jiayingzhe2022a,Liuyang2021a,Wangbiao2022a,Guoyuqian2022a,Lujianquan2018b,Zhourongpei2022a,Kobayashi2022,Yangxinrong2022,Lixi2023a,Gaoshuhua2022,Yaoyuhua2021a,Chenhongwei2020,Lifangfei2022a,Mengmin2021a,Liyiliang2022b,Lirui2021a,Chenqi2022}).
The Input-output (IO)-decoupling problem, also known as Morgan's problem \cite{MorganBS1964a}, is a popular topic that has also been recently dealt with within the BCN framework.

The IO-decoupling properties of BCNs were first characterized in \cite{ValcherME2017}, and vertex partitioning based necessary and sufficient conditions for IO-decoupling of BCNs were given in \cite{Liyifeng2022a}. The IO-decoupling definition of \cite{ValcherME2017} includes systems which are not fully output controllable. For instance, a system with two inputs $(u_1,u_2)$ and two outputs $(y_1,y_2)$, is considered IO-decoupled as per the definition of \cite{ValcherME2017} in the case where $y_1$ is controlled solely by $u_1$, and $y_2$ is not controllable by either $u_1$ or $u_2$. In the present work, we slightly modify the definition of \cite{ValcherME2017} to capture a meaningful class of fully output controllable systems where every output is controlled by a unique and distinct input. Consequently, the system in the previous illustrative example is not considered IO-decoupled as per our new definition.

In \cite{Fushihua2017}, a state feedback based IO-decoupling approach has been proposed for a particular class of BCNs that can be transformed into the IO-decomposed form, relying on solving a set of logical matrix equations. A preliminary result for solving logical matrix equations was provided in \cite{Panjinfeng2018} and used for the design of state feedback based IO-decoupling for this particular calss of BCNs. A reliable method for solving logical matrix equations was proposed in \cite{Yuyongyuan2019}. Unfortunately, the existing state feedback based IO-decoupling approaches are limited to the particular class of BCNs that can be transformed into the IO-decomposed form.
It is worth pointing out that not all BCNs can be transformed into the IO-decomposed form, and even if a BCN admits such a form, it is not a straightforward process finding it.

In the present paper, we provide some necessary and sufficient conditions, for a class of fully output controllable BCNs (that may not necessarily be transformable into an IO-decomposed form), guaranteeing the existence of a state feedback leading to a one-step\footnote{A one-step transition IO-decoupled system is defined in \cite{ValcherME2017} as an IO-decoupled system where each output at time $t+1$ is determined only by the input and output at time $t$.} transition IO-decoupled system. A constructive procedure is provided to design such a feedback, as long as the conditions are satisfied, such that the closed-loop system is IO-decoupled in the sense that every output is controlled by a unique and distinct input.
We also show that IO-decoupled BCNs, as per our new definition, can be written in a special form, referred to as the \textit{IO-decoupled canonical form}, which allows to obtain the IO-decomposed form of \cite{Fushihua2017,Panjinfeng2018} (if it exists) in a straightforward manner. The main contributions of this paper  as follows:
\begin{itemize}
  \item Two necessary and sufficient conditions for one-step transition IO-decoupling of a class of fully output controllable BCNs are derived.
  \item A new IO-decoupled canonical form, providing an explicit mapping between the decoupled inputs and outputs, is introduced.
  \item A new state feedback based approach, leading to a one-step transition IO-decoupling,  is proposed without requiring the existence of the IO-decomposed form.
\end{itemize}

The remainder of this paper is organized as follows.
Section \ref{sec2} provides the notations and the preliminaries used throughout the paper. Necessary and sufficient conditions for IO-decoupling, as well as a state feedback based IO-decoupling approach are presented in Section \ref{sec3}. A comparative analysis with respect to the existing work in the literature is given in Section \ref{sec4}.
Section \ref{sec5} provides two numerical examples illustrating the effectiveness of the proposed approach.
Finally, some concluding remarks are given in Section \ref{sec6}.

\section{Preliminaries}\label{sec2}
\subsection{Notations}\label{sec2.1}
The set of $m\times n$ real matrices is denoted by $\mathcal{M}_{m\times n}$.
The $i$-th column (row) of matrix $M$ is denoted by
$\mathrm{Col}_i(M)$ ($\mathrm{Row}_{i}(M)$).
The cardinality of a set $\Omega$ is denoted by  $|\Omega|$.
The set of $m\times n$ Boolean matrices is denoted by
$\mathcal{B}_{m\times n}:=\{B\in\mathcal{M}_{m\times n}|B_{ij}\in\mathcal{D}\}$, where $\mathcal{D}=\{0,1\}$.
The set of $m\times n$ logical matrices is denoted by $\mathcal{L}_{m\times n}:=\{L\in\mathcal{B}_{m\times n}|\mathrm{Col}_{i}(L)\in\Delta_{m},i=1,\ldots,n\}$, where $\Delta_{m}=\{\delta_{m}^{i}|\delta_{m}^{i}=\mathrm{Col}_{i}(I_{m}),i=1,\ldots,m\}$ and $I_{m}=\mathrm{diag}\{1,\ldots,1\},m\geq 2$.
We denote $\Delta$ by $\Delta_{2}$ and
the logical matrix $[\delta_{m}^{i_{1}}\ \cdots\ \delta_{m}^{i_{n}}]$ by $\delta_{m}[i_{1}\ \cdots\ i_{n}]$.
We define the scalar-valued function $\mathrm{sgn}(a)$, which returns $0$ if $a=0$ and $1$ if $a\neq 0$.
For a given matrix $M=[M_{ij}]_{p\times q}$, we define  $\mathrm{sgn}(M)=[\mathrm{sgn}(M_{ij})]_{p\times q}$.
We define $V(\Omega):=\delta_n^{i_1}+\delta_n^{i_2}+\cdots+\delta_n^{i_m}$, where  $\Omega=\{\delta_n^{i_1},\delta_n^{i_2},\ldots,\delta_n^{i_m}\}$. The Kronecker product and Hadamard product are denoted by $\otimes$ and $\circ$, respectively.
Given $k$ matrices  $M_{1},\ldots,M_{k}\in\mathcal{M}_{m\times n}$, the Hadamard product of $M_{1},\ldots, M_{k}$ is defined as $\mathop{\prod_{H}}\limits_{i=1\ \ }^{k\ \ }M_{i}=M_{1}\circ\cdots\circ M_{k}$.
The swap matrix with indices $n$ and $m$ is defined as
$W_{[n,m]}=[I_{m}\otimes\delta_{n}^{1}\ \cdots\ I_{m}\otimes\delta_{n}^{n}]$.
The $n$-dimensional vectors of ones and zeros are denoted by  $\mathbf{1}_{n}=\sum\limits_{i=1}^{n}\delta_{n}^{i}$ and $\mathbf{0}_{n}=[0\ \cdots\ 0]^{\top}$, respectively.

\subsection{Semi-tensor product of matrices}\label{subsec2.1}
Given two matrices $P\in\mathcal{M}_{m\times n}$ and $Q\in\mathcal{M}_{p\times q}$,
the semi-tensor product of $P$ and $Q$ is defined as
\[P\ltimes Q:=(P\otimes I_{t/n})(Q\otimes I_{t/p}),\]
where $t$ is the least common multiple of $n$ and $p$.
In this paper, the default matrix product is the semi-tensor product and the symbol $\ltimes$ is omitted.
Define a one-to-one mapping $\varphi$ from $\mathcal{D}$ to $\Delta$ such that  $\varphi(1)=\delta_{2}^{1}$ and $\varphi(0)=\delta_{2}^{2}$.
For any logical variable $X_{i}\in\mathcal{D},i=1,\ldots,n$ and any $n$-ary Boolean function $f(X_{1},\ldots,X_{n})$, one has the following vector forms:
\[
\varphi(X_{i})=\left[
\begin{array}{c}
X_{i}\\
1-X_{i}\\
\end{array}
\right]
\]
and
\[
\varphi(f(X_{1},\ldots,X_{n}))=\left[
\begin{array}{c}
f(X_{1},\ldots,X_{n})\\
1-f(X_{1},\ldots,X_{n})\\
\end{array}
\right]_{,}
\]
where $i=1,\ldots,n$.
The following lemma shows that for any Boolean function, there is a structure matrix such that its vector form can be expressed via a linear form.

\begin{lemma}\label{le1}\rm\cite{Chengdaizhan2011b}
For an $n-$ary Boolean function $f(X_{1},\ldots,X_{n})$, there exists a unique structure matrix $L_{f}\in\mathcal{L}_{2\times 2^{n}}$ such that the vector form of $f(X_{1},\ldots,X_{n})$ is expressed as
\[ \varphi(f(X_{1},\ldots,X_{n}))=L_{f}\ltimes_{i=1}^{n}x_{i},\]
where $x_{i}=\varphi(X_{i}),i=1,\ldots,n$.
\end{lemma}

In addition, the following lemma is necessary for the introduction of BCN algebraic forms.

\begin{lemma}\label{le2}\rm\cite{Chengdaizhan2011b}
Assume
\[
\left\{
\begin{array}{c}
y=M_{y}\ltimes_{i=1}^{n}x_{i},\\
z=M_{z}\ltimes_{i=1}^{n}x_{i},\\
\end{array}
\right.
\]
where $x_{i},y,z\in\Delta,i=1,2,\ldots,n,M_{y},M_{z}\in\mathcal{L}_{2\times 2^{n}}$.
Then
\[
yz=(M_{y}\ast M_{z})\ltimes_{i=1}^{n}x_{i},
\]
where
$M_{y}\ast M_{z}=[\mathrm{Col}_{1}(M_{y})\otimes \mathrm{Col}_{1}(M_{z})\ \cdots\ \mathrm{Col}_{2^{n}}(M_{y})\otimes \mathrm{Col}_{2^{n}}(M_{z})].$
\end{lemma}

\subsection{Algebraic forms of BCNs}\label{subsec2.2}
Consider the following BCN:
\begin{equation}\label{e1}
\left\{
\begin{array}{l}
X_{i}(t+1)=f_{i}(X(t),U(t)),\\
Y_{j}(t)=g_{j}(X(t)),\\
\end{array}
\right.
\end{equation}
where $X(t)=(X_{1}(t),\ldots,X_{n}(t)), U(t)=(U_{1}(t),\ldots,$
$U_{m}(t))$ and $Y(t)=(Y_{1}(t),\ldots,Y_{p}(t))$ are the state, input and output of the system, respectively, $f_{i}$ and $g_{j}$ are Boolean functions.

Let $x_{i}=\varphi(X_{i}),u_{i}=\varphi(U_{i})$ and $y_{i}=\varphi(Y_{i})$.
From Lemma \ref{le1} and Lemma \ref{le2}, BCN (\ref{e1}) is transformed into the following algebraic form:
\begin{equation}\label{e2}
\left\{
\begin{array}{l}
x(t+1)=Lu(t)x(t),\\
y_{i}(t)=H_{i}x(t),\\
\end{array}
\right.
\end{equation}
where $x(t)=\ltimes_{i=1}^{n}x_{i}(t)\in\Delta_{2^{n}}$ and $u(t)=\ltimes_{i=1}^{m}u_{i}(t)\in\Delta_{2^{m}}$ are vector forms of $X(t)$ and $U(t)$, respectively,
$y_{i}(t)\in\Delta$, $L\in\mathcal{L}_{2^{n}\times2^{m+n}}$ and $H_{i}\in\mathcal{L}_{2\times 2^{n}},i=1,\ldots,p$.

\section{Main results}\label{sec3}
In this section, necessary and sufficient conditions for BCN IO-decoupling are derived. Thereafter, the canonical form for an IO-decoupled BCN is established. Finally, a state feedback control design approach is proposed for BCN IO-decoupling.

\subsection{Characterization of the IO-decoupling}\label{subsec3.1}

The one-step transition IO-decoupling is defined as follows.

\begin{definition}\label{def2}\rm
BCN (\ref{e2}) with inputs and outputs having the same dimension (\ie, $m=p$) is said to be one-step transition IO-decoupled if for every $t \in \mathbb{Z}_+$, every index $i\in\{1,\ldots,m\}$, and every pair of states $x(t),\hat{x}(t)\in\Delta_{2^{n}}$ satisfying $y_i(t)=\hat{y}_i(t)$, with $y_i(t):=H_ix(t)$ and $\hat{y}_i(t):=H_i\hat{x}(t)$, one has
\[
y_i(t+1)=\hat{y}_i(t+1)~\iff~u_i(t)=\hat{u}_i(t)\\
\]
for every pair of inputs $u,\hat{u}\in\Delta_{2^{m}}$, where $u_i$ and $\hat{u}_i$ are the $i$-th entries of $u$ and $\hat{u}$ respectively.
\end{definition}

\begin{remark}\rm\label{re3.1.1}
Note that Definition \ref{def2} is slightly different from the definition in \cite{ValcherME2017}. The difference is that the one-step transition IO-decoupling property in Definition \ref{def2} imposes that every output is controlled by a unique input.
\end{remark}

In the sequel, we assume that the inputs and outputs of BCN \eqref{e2} have the same dimension (\ie, $m=p$).
We include all states with $y_{i}=\delta_{2}^{j}$ into the set $\Gamma_{ij}=\{x|H_{i}x=\delta_{2}^{j}\}$,
and all inputs with $u_{i}=\delta_{2}^{j}$ into the set $\Omega_{ij}=\{u|u=\delta_{2^{i-1}}^{k_{1}}\ltimes\delta_{2}^{j}\ltimes\delta_{2^{m-i}}^{k_{2}},k_{1}\in\{1,\ldots,2^{i-1}\},k_{2}\in\{1,\ldots,2^{m-i}\}\}$, where $i\in\{1,\ldots,m\}$ and $j=1,2$.

According to Definition \ref{def2}, for a given state $x\in\Delta_{2^{n}}$, a one-step transition IO-decoupled BCN satisfies:
\[
\left\{
\begin{array}{ll}
H_{i}Lux=H_{i}L\hat{u}x,&\mathrm{if}\ u,\hat{u}\in\Omega_{ik},\\
H_{i}Lux\neq H_{i}L\hat{u}x,&\mathrm{if}\ u\in\Omega_{i1},\hat{u}\in\Omega_{i2},\\
\end{array}
\right.
\]
where $i\in\{1,\ldots,m\},k=1,2$.
As a result of $H_{i}Lux,H_{i}L\hat{u}x\in\Delta$, one has
\[
\left\{
\begin{array}{l}
H_{i}LV(\Omega_{ik})x\in\{2^{m-1}\delta_{2}^{1},2^{m-1}\delta_{2}^{2}\},\\
H_{i}LV(\Omega_{i1})x\neq H_{i}LV(\Omega_{i2})x,\\
\end{array}
\right.
\]
where $i\in\{1,\ldots,m\},k=1,2$.
On the other hand, Definition \ref{def2} implies that for all inputs $u\in\Delta_{2^{m}}$ with the same $u_{i}$, $y_i=H_iLux$ remains unchanged  for any state $x\in\Gamma_{ij},j=1,2$,
which requires
\[
\sum\limits_{x\in\Gamma_{ij}}H_{i}LV(\Omega_{ik})x\in\{2^{m-1}|\Gamma_{ij}|\delta_{2}^{1},2^{m-1}|\Gamma_{ij}|\delta_{2}^{2}\},
\]
where $j,k=1,2$.
Therefore, we obtain  \[\mathrm{sgn}(\sum\limits_{x\in\Gamma_{ij}}H_{i}LV(\Omega_{ik})x)\in\Delta,j,k=1,2.\]
In light of the above analysis, if a given BCN is one-step transition IO-decoupled, then for each index $i\in\{1,\ldots,m\}$, the following conditions hold:
\begin{equation}\label{e3.1.3}
\left\{
\begin{array}{l}
\mathrm{sgn}(\sum\limits_{x\in\Gamma_{ij}}H_{i}LV(\Omega_{ik})x)\in\Delta,\\
\mathrm{sgn}(\sum\limits_{x\in\Gamma_{ij}}H_{i}LV(\Omega_{i1})x)
\neq\mathrm{sgn}(\sum\limits_{x\in\Gamma_{ij}}H_{i}LV(\Omega_{i2})x),
\end{array}
\right.
\end{equation}
where $j,k=1,2$.

Next, we construct such a $4\times 2$ Boolean matrix $\Xi^{i}$ expressed as follows:
\begin{align}\label{e3}
\Xi^{i}=\left[
\begin{array}{cc}
\Xi^{i}_{11}&\Xi^{i}_{12}\\
\Xi^{i}_{21}&\Xi^{i}_{22}\\
\end{array}
\right]
\end{align}
where $\Xi^{i}_{jk}=\mathrm{sgn}(\sum\limits_{x\in\Gamma_{ij}}H_{i}LV(\Omega_{ik})x),j,k=1,2$ and $i\in\{1,\ldots,m\}$.
For each index $i\in\{1,\ldots,m\}$ and $j,k\in\{1,2\}$, $\Xi^{i}_{jk}$ contains all possible values of $y_i(t+1)$ when
$u_{i}(t)=\delta_{2}^{k}$ and $y_{i}(t)=\delta_{2}^{j}$.
Combined with \eqref{e3.1.3}, if a given BCN is one-step transition IO-decoupled, then for each index $i\in\{1,\ldots,m\}$, the constructed Boolean matrix $\Xi^{i}$ satisfies
\begin{equation}\label{e3.1.4}
\left\{
\begin{array}{l}
(\delta_{2}^{j})^{\top}\Xi^{i}\delta_{2}^{k}\in\Delta,\\
(\delta_{2}^{j})^{\top}\Xi^{i}\delta_{2}^{1}\neq(\delta_{2}^{j})^{\top}\Xi^{i}\delta_{2}^{2},\\
\end{array}
\right.
\end{equation}
where $j,k=1,2$.

Based on the analysis above, a necessary and sufficient condition for one-step transition IO-decoupling is proposed using the constructed matrices $\Xi^{1},\ldots,\Xi^{m}$.

\begin{theorem}\label{th1}\rm
BCN (\ref{e2}), with inputs and outputs having the same dimension (\ie, $m=p$), is one-step transition IO-decoupled, as per Definition \ref{def2}, if and only if for each index $i\in\{1,\ldots,m\}$, the constructed matrix $\Xi^{i}$ satisfies
\begin{equation}\label{e3.1.5}
(\delta_{2}^{j})^{\top}\Xi^{i}\mathbf{1}_{2}=\mathbf{1}_{2},j=1,2.
\end{equation}
\end{theorem}

\begin{proof}
(Necessity) If BCN (\ref{e2}) is one-step transition IO-decoupled as per Definition \ref{def2}, then one has  (\ref{e3.1.4}).
It follows that, for each $j=1,2$, if $(\delta_2^j)^{\top}\Xi^i\delta_2^1=\delta_2^1$, then $(\delta_2^j)^{\top}\Xi^i\delta_2^2=\delta_2^2$,
otherwise, $(\delta_2^j)^{\top}\Xi^i\delta_2^2=\delta_2^1$.
Thus,
\[
\begin{array}{rcl}
(\delta_{2}^{j})^{\top}\Xi^{i}\mathbf{1}_{2}&=&(\delta_{2}^{j})^{\top}\Xi^{i}(\delta_{2}^{1}+\delta_{2}^{2})\\
\ &=&(\delta_{2}^{j})^{\top}\Xi^{i}\delta_{2}^{1}+(\delta_{2}^{j})^{\top}\Xi^{i}\delta_{2}^{2}\\
\ &=&\mathbf{1}_{2},
\end{array}
\]
where $j=1,2$.

(Sufficiency) Since (\ref{e3.1.5}) is equivalent to $(\delta_{2}^{j})^{\top}\Xi^{i}\delta_{2}^{1}$
$+(\delta_{2}^{j})^{\top}\Xi^{i}\delta_{2}^{2}$
$=\mathbf{1}_{2}$, one of three following cases may occur:
\begin{enumerate}
  \item $(\delta_{2}^{j})^{\top}\Xi^{i}\delta_{2}^{k}\in\Delta$ and $(\delta_{2}^{j})^{\top}\Xi^{i}\delta_{2}^{1}\neq(\delta_{2}^{j})\Xi^{i}\delta_{2}^{2}$,
  \item $(\delta_{2}^{j})^{\top}\Xi^{i}\delta_{2}^{1}=\mathbf{1}_{2}$ and $(\delta_{2}^{j})^{\top}\Xi^{i}\delta_{2}^{2}=\mathbf{0}_{2}$,
  \item $(\delta_{2}^{j})^{\top}\Xi^{i}\delta_{2}^{1}=\mathbf{0}_{2}$ and $(\delta_{2}^{j})^{\top}\Xi^{i}\delta_{2}^{2}=\mathbf{1}_{2}$.
\end{enumerate}
According to the construction of $\Xi^{i}$, it is impossible to obtain  \[(\delta_{2}^{j})^{\top}\Xi^{i}\delta_{2}^{k}=\mathbf{0}_{2},\]
where $j,k=1,2$.
Thus, (\ref{e3.1.4}) can be inferred from (\ref{e3.1.5}).
For each index $i\in\{1,\ldots,m\}$ and every pair of states $x,\hat{x}\in\Delta_{2^{n}}$ satisfying $H_ix=H_i\hat{x}$ (assume $H_ix=H_i\hat{x}=\delta_2^j$), if $u,\hat{u}\in\Omega_{ik},k=1,2$, then $H_iLux=(\delta_2^j)^{\top}\Xi^i\delta_2^k=H_iL\hat{u}\hat{x}$.
If $u\in\Omega_{i1},\hat{u}\in\Omega_{i2}$, then one has $H_iLux=(\delta_2^j)^{\top}\Xi^{i}\delta_2^1$ and $H_iL\hat{u}\hat{x}=(\delta_2^j)^{\top}\Xi^{i}\delta_2^2$, which means $H_iLux\neq H_iL\hat{u}\hat{x}$.
Therefore, BCN (\ref{e2}) is one-step transition IO-decoupling as per Definition \ref{def2}.
\end{proof}

It is worth pointing out that if a BCN is IO-decoupled as per Definition \ref{def2}, then an explicit mapping between outputs and inputs can be given. To show that, let us consider $u_{i}(t)=\delta_{2}^{k},k=1,2$ and $y_{i}(t)=\delta_{2}^{j},j=1,2$. By virtue of Theorem \ref{th1}, one has $(\delta_2^j)^{\top}\Xi^i\delta_2^k\in\Delta$ since (\ref{e3.1.4}) and  (\ref{e3.1.5}) are equivalent. Consequently,
\[
\begin{array}{rcl}
y_{i}(t+1)&=&(\delta_{2}^{j})^{\top}\Xi^{i}\delta_{2}^{k}\\
\ &=&(y_{i}(t))^{\top}\Xi^{i}u_{i}(t)\\ &=&[(\delta_{2}^{1})^{\top}\Xi^{i}\ (\delta_{2}^{2})^{\top}\Xi^{i}]y_{i}(t)u_{i}(t).\\
\end{array}
\]

\begin{corollary}\label{pro01}\rm
BCN (\ref{e2}), with inputs and outputs having the same dimension (\ie, $m=p$), is one-step transition IO-decoupled, as per Definition \ref{def2},
if and only if for each index $i\in\{1,\ldots,m\}$, the mapping between $y_{i}(t+1)$ and $u_{i}(t),y_{i}(t)$ can be expressed as follows:
\begin{equation}\label{e01}
y_{i}(t+1)=[(\delta_{2}^{1})^{\top}\Xi^{i}\ (\delta_{2}^{2})^{\top}\Xi^{i}]y_{i}(t)u_{i}(t),
\end{equation}
where $\Xi^i$ satisfies (\ref{e3.1.5}).
\end{corollary}

\begin{remark}\label{re5}\rm
It is worth noting that $(\delta_{2}^{j})^{\top}\Xi^{i}\mathbf{1}_2=\mathbf{1}_2$ implies that $(\delta_{2}^{j})^{\top}\Xi^{i}\delta_{2}^{k}\in\Delta$ and $(\delta_{2}^{j})^{\top}\Xi^{i}\delta_{2}^{1}\neq(\delta_{2}^{j})^{\top}\Xi^{i}\delta_{2}^{2}$.
The former one illustrates the fact that the output $y_i$ is unaffected by inputs $u_1,\ldots,u_{i-1},u_{i+1},\ldots,u_{m}$.
The latter one shows that the output $y_i$ is uniquely controlled by $u_i$.
Unfortunately, the condition $(\delta_{2}^{j})^{\top}\Xi^{i}\delta_{2}^{1}\neq(\delta_{2}^{j})^{\top}\Xi^{i}\delta_{2}^{2}$ is not reflected in (\ref{e01}), and as such, \eqref{e01} does not imply that the system is one-step transition IO-decoupled as per Definition \ref{def2}.
However, as it is going to be shown later, it is a necessary and sufficient condition for the one-step transition IO-decoupling as per the definition in  \cite{ValcherME2017}.
\end{remark}

\subsection{The canonical form}\label{subsec3.3}

Next, we will provide a canonical form for the class of one-step transition IO-decoupled BCNs. Let us assume that BCN (\ref{e2}) is one-step transition IO-decoupled, and let $z_{i}(t)=y_{i}(t),i\in\{1,\ldots,m\}$.
According to Corollary \ref{pro01}, one has
\begin{equation}\label{e3.1.1}
\left\{
\begin{array}{l}
z_{i}(t+1)=[(\delta_{2}^{1})^{\top}\Xi^{i}\ (\delta_{2}^{2})^{\top}\Xi^{i}]W_{[2,2]}u_{i}(t)z_{i}(t),\\
y_{i}(t)=z_{i}(t),\\
\end{array}
\right.
\end{equation}
where $i\in\{1,\ldots,m\}$.
BCN \eqref{e3.1.1} is referred to as the one-step transition IO-decoupled canonical form if (\ref{e3.1.5}) holds for $i\in\{1,\ldots,m\}$.

From \eqref{e3.1.1}, one can notice that
for each index $i\in\{1,\ldots,m\}$, each initial state $z_{i}(0)$ can evolve to any state within $\Delta$ under an appropriate control sequence.
Due to the independence of the $z_{i}$ subsystems, one can conclude that each initial state $\bar{z}(0)=\ltimes_{i=1}^{m}z_{i}(0)\in\Delta_{2^{m}}$ can evolve to any state within $\Delta_{2^{m}}$ under an appropriate control sequence. Therefore, system (\ref{e3.1.1}) is controllable. For a BCN that possesses a one-step transition IO-decoupled canonical form \eqref{e3.1.1}, one has $\bar{z}(t)=\ltimes_{i=1}^{m}z_i(t)=\ltimes_{i=1}^{m}y_i(t)=(H_1\ast\cdots\ast H_m)x(t)$ due to the fact that $z_i(t)=y_i(t),i\in\{1,\ldots,m\}$.
Therefore, for each initial state $x(0)\in\Delta_{2^{n}}$, a unique $\bar{z}(0)$ can be found such that $\bar{z}(0)=(H_1\ast\cdots\ast H_m)x(0)$ holds. This indicates that each initial state $x(0)\in\Delta_{2^{n}}$ can evolve to any output within $\Delta_{2^{m}}$ under an appropriate control sequence, which means that the BCN under consideration is output controllable.

In the following, we use the one-step transition IO-decoupled canonical form to obtain the IO-decomposed form of \cite{Fushihua2017,Panjinfeng2018} in a straightforward manner.

Consider BCN \eqref{e2} that possesses a one-step transition IO-decoupled canonical form \eqref{e3.1.1}.
Take $z_{m+1}(t)=P_{m+1}x(t)$, where $z_{m+1}(t)\in\Delta_{2^{n-m}}$ and  $P_{m+1}\in\mathcal{L}_{2^{n-m}\times2^{n}}$.
Define $z(t):=\ltimes_{i=1}^{m+1}z_{i}(t)$, where $z(t)\in\Delta_{2^{n}}$. It follows that $z(t)=(H_{1}\ast\cdots\ast H_{m}\ast P_{m+1})x(t):=Px(t)$,
where  $P=H_{1}\ast\cdots\ast H_{m}\ast P_{m+1}\in\mathcal{L}_{2^{n}\times 2^{n}}$.
If $P$ is nonsingular, then BCN (\ref{e2}) can be transformed into the following IO-decomposed form:
\begin{equation}\label{e3.1.2}
\left\{
\begin{array}{l}
z_{i}(t+1)=[(\delta_{2}^{1})^{\top}\Xi^{i}\ (\delta_{2}^{2})^{\top}\Xi^{i}]W_{[2,2]}u_{i}(t)z_{i}(t),\\
z_{m+1}(t+1)=P_{m+1}L(I_{2^{m}}\otimes P^{\top})u(t)z(t),\\
y_{i}(t)=z_{i}(t),\\
\end{array}
\right.
\end{equation}
where $i\in\{1,\ldots,m\}$.
The following theorem provides necessary and sufficient conditions under which  a one-step transition IO-decoupled BCN admits an IO-decomposed form.

\begin{theorem}\label{th5}\rm
If BCN (\ref{e2}), with inputs and outputs having the same dimension (\ie, $m=p$), is one-step transition IO-decoupled, as per Definition \ref{def2},  then it has an IO-decomposed form (\ref{e3.1.2}) if and only if
\[
(H_{1}\ast\cdots\ast H_{m})\mathbf{1}_{2^{n}}=2^{n-m}\mathbf{1}_{2^{m}}.
\]
\end{theorem}
\begin{proof}
(Sufficiency) Since BCN (\ref{e2}) is one-step transition IO-decoupled as per Definition \ref{def2}, it can be transformed into the canonical form (\ref{e3.1.1}).
Let $\Lambda_{j}=\{x|(H_{1}\ast\cdots\ast H_{m})x=\delta_{2^{m}}^{j}\},j\in\{1,\ldots,2^{m}\}$.
Then, $\Lambda_{i}\cap\Lambda_{j}=\emptyset,i\neq j$ and $\Lambda_{1}\cup\cdots\cup\Lambda_{2^{m}}=\Delta_{2^{n}}$.
Due to the fact that $(H_{1}\ast\cdots\ast H_{m})\mathbf{1}_{2^{n}}=2^{n-m}\mathbf{1}_{2^{m}}$, one has $|\Lambda_{j}|=2^{n-m}$, $j\in\{1,\ldots,2^{m}\}$.
For $j\in\{1,\ldots,2^{m}\}$, denote $|\Lambda_{j}|=\{\delta_{2^{n}}^{\lambda_{j1}},\ldots,\delta_{2^{n}}^{\lambda_{j|\Lambda_{j}|}}\}$.
Define a one-to-one mapping $\psi$ from $\Delta_{2^{n}}$ to $\Delta_{2^{n}}$ such that
$\psi(\delta_{2^{n}}^{\lambda_{jk}})=\delta_{2^{m}}^{j}\ltimes\delta_{2^{n-m}}^{k}$, where $\delta_{2^{n}}^{\lambda_{jk}}\in\Lambda_{j},j\in\{1,\ldots,2^{m}\}$ and $k\in\{1,\ldots,2^{n-m}\}$.
It is easy to express $\psi(x)$ as $\psi(x)=Px$, where $P\in\mathcal{L}_{2^{n}\times 2^{n}}$ is nonsingular.
In addition, two surjections can be obtained from $\psi$:
\begin{enumerate}
\item $\psi_1$ is a surjection from $\Delta_{2^{n}}$ to $\Delta_{2^{m}}$ such that $\psi_1(\delta_{2^{n}}^{\lambda_{jk}})=\delta_{2^{m}}^{j}$, where $\delta_{2^{n}}^{\lambda_{jk}}\in\Lambda_{j},j\in\{1,\ldots,2^{m}\}$ and $k\in\{1,\ldots,2^{n-m}\}$.
\item $\psi_2$ is a surjection from $\Delta_{2^{n}}$ to $\Delta_{2^{n-m}}$ such that $\psi_2(\delta_{2^{n}}^{\lambda_{jk}})=\delta_{2^{n-m}}^{k}$, where $\delta_{2^{n}}^{\lambda_{jk}}\in\Lambda_{j},j\in\{1,\ldots,2^{m}\}$ and $k\in\{1,\ldots,2^{n-m}\}$.
\end{enumerate}
One can also express $\psi_{i}(x),i=1,2$ as $\psi_{1}(x)=(H_{1}\ast\cdots\ast H_{m})x$ and $\psi_{2}(x)=P_{m+1}x$, where $P_{m+1}\in\mathcal{L}_{2^{n-m}\times 2^{n}}$.
For each $j\in\{1,\ldots,2^{m}\}$ and $k\in\{1,\ldots,2^{n-m}\}$, if $x=\delta_{2^{n}}^{\lambda_{jk}}\in\Lambda_{j}$, then one has
\[
\begin{array}{rcl}
\psi(x)&=&\delta_{2^{m}}^{j}\ltimes\delta_{2^{n-m}}^{k}\\
\ &=&\psi_{1}(x)\ltimes\psi_{2}(x)\\
\ &=&((H_{1}\ast\cdots\ast H_{m})x)\ltimes(P_{m+1}x)\\
\ &=&(H_{1}\ast\cdots\ast H_{m}\ast P_{m+1})x.\\
\end{array}
\]
This implies $P=H_{1}\ast\cdots\ast H_{m}\ast P_{m+1}$.
Taking $z=\psi(x)=Px$, the IO-decomposed form (\ref{e3.1.2}) is obtained.

(Necessity) If BCN (\ref{e2}) has an IO-decomposed form (\ref{e3.1.2}), then $P=H_{1}\ast\cdots\ast H_{m}\ast P_{m+1}$ is nonsingular.
Defining $\psi_{1}(x)=(H_{1}\ast\cdots\ast H_{m})x$ and $\psi_{2}(x)=P_{m+1}x$, one has $z=Px=\psi_{1}(x)\ltimes\psi_{2}(x)$.
Assume $z=\delta_{2^{n}}^{\alpha}\in\Delta_{2^{n}}$.
Due to the uniqueness of  $\delta_{2^{n}}^{\alpha}=\delta_{2^{m}}^{\alpha_{1}}\ltimes\delta_{2^{n-m}}^{\alpha_{2}}$,
$\psi_{1}$ and $\psi_{2}$ are surjections from $\Delta_{2^{n}}$ to $\Delta_{2^{m}}$ and $\Delta_{2^{n-m}}$ respectively.
For each possible value of $\psi_{1}(x)$ (assume $\psi_{1}(x)=\delta_{2^{m}}^{\alpha_{1}}$),
taking $\psi_{2}(x)$ from $\delta_{2^{n-m}}^{1}$ to $\delta_{2^{n-m}}^{2^{n-m}}$,
$2^{n-m}$ different $z$ can be found.
It follows that, there are $2^{n-m}$ states with $\psi_{1}(x)=\delta_{2^{m}}^{\alpha_{1}}$. 
In other words, $2^{n-m}$ columns of $H_{1}\ast\cdots\ast H_{m}$ equal to $\delta_{2^{m}}^{\alpha_{1}}$.
Due to the arbitrariness of $\delta_{2^{n}}^{\alpha_{1}}$, it can be concluded that $(H_{1}\ast\cdots\ast H_{m})\mathbf{1}_{2^{n}}=2^{n-m}\mathbf{1}_{2^{m}}$.
\end{proof}

It is worth pointing out that, in view of Theorem \ref{th5}, not all one-step transition IO-decoupled BCNs can be transformed into an IO-decomposed form, which exists when $(H_{1}\ast\cdots\ast H_{m})\mathbf{1}_{2^{n}}=2^{n-m}\mathbf{1}_{2^{m}}$ is satisfied.

\subsection{State feedback control design}\label{subsec3.2}
In this section, we consider a BCN which is not IO-decoupled and provide some condition under which there exists a state feedback that decouples the BCN. In the sequel, we do not require that the inputs and outputs have the same dimension.

Consider the following state feedback with a new input $v(t)$:
\begin{equation}\label{e3.2.1}
u(t)=Kx(t)v(t),
\end{equation}
where $K\in\mathcal{L}_{2^{m}\times 2^{n+p}}$.
Substituting (\ref{e3.2.1}) into BCN (\ref{e2}), one obtains the following closed-loop BCN with a new input $v(t)$:
\begin{equation}\label{e3.2.2}
\left\{
\begin{array}{l}
x(t+1)=L(Kx(t)v(t))x(t),\\
y_{i}(t)=H_{i}x(t),\\
\end{array}
\right.
\end{equation}
where $i\in\{1,\ldots,p\}$.
The designed state feedback aims at making the BCN (\ref{e3.2.2}) one-step transition IO-decoupled.

For each output node $y_{i}$, a new Boolean matrix $\widehat{\Xi}^{i}$ with the same meaning as $\Xi^{i}$ is written as
\begin{align}\label{e3.3.5}
\widehat{\Xi}^{i}=\left[
\begin{array}{cc}
\widehat{\Xi}^{i}_{11}&\widehat{\Xi}^{i}_{12}\\
\widehat{\Xi}^{i}_{21}&\widehat{\Xi}^{i}_{22}\\
\end{array}
\right]_{,}
\end{align}
where $\widehat{\Xi}^{i}_{lk}=\mathrm{sgn}(\sum\limits_{x\in\Gamma_{il}}H_{i}L(KxV(\hat{\Omega}_{ik}))x),l,k=1,2,i\in\{1,\ldots,p\}$ and all inputs with $v_{i}=\delta_{2}^{j},i\in\{1,\ldots,p\}$ are included in the set $\hat{\Omega}_{ij}=\{v|v=\delta_{2^{i-1}}^{k_{1}}\ltimes\delta_{2}^{j}\ltimes\delta_{2^{p-i}}^{k_{2}},k_{1}\in\{1,\ldots,2^{i-1}\},$
$k_{2}\in\{1,\ldots,2^{p-i}\}\},j=1,2$.
Similarly, $\widehat{\Xi}^{i}_{lk}$ contains all possible values of $y_{i}(t+1)$ when $v_{i}(t)=\delta_{2}^{k}$ and $y_{i}(t)=\delta_{2}^{l}$.
According to Theorem \ref{th1}, BCN (\ref{e3.2.2}) is one-step transition IO-decoupled if and only if for each index $i\in\{1,\ldots,p\}$, the new Boolean matrix $\widehat{\Xi}^{i}$ satisfies
\begin{equation}\label{e3.3.4}
(\delta_{2}^{l})^{\top}\widehat{\Xi}^{i}\mathbf{1}_{2}=\mathbf{1}_{2},l=1,2.
\end{equation}
That is, the designed state feedback matrix $K$ should result in
\begin{equation}\label{e3.3.3}
\left\{
\begin{array}{l}
(\delta_{2}^{l})^{\top}\widehat{\Xi}^{i}\delta_{2}^{k}\in\Delta,\\
(\delta_{2}^{l})^{\top}\widehat{\Xi}^{i}\delta_{2}^{1}\neq(\delta_{2}^{l})^{\top}\widehat{\Xi}^{i}\delta_{2}^{2},\\
\end{array}
\right.
\end{equation}
where $i\in\{1,\ldots,p\},k,l=1,2$.
Therefore, we design the state feedback matrix $K$ such that $y_{i}(t+1)=H_{i}L(Kx(t)v(t))x(t)=v_{i}(t)$.
Based on this, for each index $i\in\{1,\ldots,p\}$, an auxiliary matrix, denoted by $\hat{K}_{i}\in\mathcal{L}_{2\times 2^{n+m}}$, is constructed as follows:
\begin{equation}\label{e3.2.4}
\mathrm{Col}_{k_{\alpha}}(\hat{K}_{i}\delta_{2^{n}}^{\alpha})=\delta_{2}^{j}\ \mathrm{if}\ \mathrm{sgn}(\sum\limits_{\delta_{2^{n}}^{\alpha}\in\Gamma_{il}}H_{i}L\delta_{2^{m}}^{k_{\alpha}}\delta_{2^{n}}^{\alpha})=\delta_{2}^{j},
\end{equation}
where $\alpha\in\{1,\ldots,2^{n}\},j=1,2$.
If the state feedback control (\ref{e3.2.1}) leads to $y_i(t+1)=v_i(t),i\in\{1,\ldots,p\}$, then for each state $x\in\Delta_{2^{n}}$, there is an input $v$ with $v_i=\delta_2^j$ and a column of $\hat{K}_ix$ with $\mathrm{Col}_{k_{\alpha}}(\hat{K}_ix)=\delta_2^j$  such that $\delta_{2^{m}}^{k_{\alpha}}=Kxv$.
Thus, for each index $i\in\{1,\ldots,p\}$ and each state $x\in\Delta_{2^{n}}$, let
\begin{equation}\label{e3.3.1}
KxV(\hat{\Omega}_{ij})=(\mathrm{Row}_{j}(\hat{K}_{i}x))^{\top},
\end{equation}
where $j=1,2$.
BCN (\ref{e3.2.2}) is one-step transition IO-decoupled if no zero column exists in $K$ generated from (\ref{e3.3.1}).
Consequently, a necessary and sufficient condition for one-step transition IO-decoupling via state feedback is provided in the following theorem:

\begin{theorem}\label{th3}\rm
BCN (\ref{e2}) is one-step transition IO-decoupled (as per Definition \ref{def2}) under the state feedback control (\ref{e3.2.1}) if and only if for any $\alpha\in\{1,\ldots,2^{n}\}$, one has
\begin{equation}\label{e3.2.3}
\left\{
\begin{array}{l}
\mathrm{sgn}(\hat{K}_{i}\delta_{2^{n}}^{\alpha}\mathbf{1}_{2^{m}})=\mathbf{1}_{2},\\
\mathop{\prod\nolimits_{H}}\limits_{i=1\ \ }^{p\ \ }((\hat{K}_{i}\delta_{2^{n}}^{\alpha})^{\top}\delta_{2}^{\eta_{i}})\neq\mathbf{0}_{2^{m}},\\
\end{array}
\right.
\end{equation}
where $\eta_{i}=1,2,i\in\{1,\ldots,p\}$.
\end{theorem}

\begin{proof}
(Necessity) If BCN (\ref{e2}) is one-step transition IO-decoupled under the state feedback control (\ref{e3.2.1}), \ie, BCN (\ref{e3.2.2}) is one-step transition IO-decoupled, then for each index $i\in\{1,\ldots,p\}$, the constructed Boolean matrix $\widehat{\Xi}^{i}$ satisfies (\ref{e3.3.4}), which means that (\ref{e3.3.3}) holds for $k,l=1,2$.
For a given index $i\in\{1,\ldots,p\}$,
assume \[(\delta_{2}^{l})^{\top}\widehat{\Xi}^{i}\delta_{2}^{k}=\mathrm{sgn}(\sum\limits_{\delta_{2^{n}}^{\alpha}\in\Gamma_{il}}H_{i}L(K\delta_{2^{n}}^{\alpha}V(\hat{\Omega}_{ik}))\delta_{2^{n}}^{\alpha})=\delta_{2}^{j}.\]
For each state $x\in\Delta_{2^{n}}$
(assume $x=\delta_{2^{n}}^{\alpha}$) and all possible $\delta_{2^{m}}^{k_{\alpha}},\ k_{\alpha}\in\{1,\ldots,2^{m}\}$ satisfying $(\delta_{2^{m}}^{k_{\alpha}})^{\top}(K\delta_{2^{n}}^{\alpha}V(\hat{\Omega}_{ik}))\neq 0$, one has
$\mathrm{sgn}(\sum\limits_{\delta_{2^{n}}^{\alpha}\in\Gamma_{il}}H_iL\delta_{2^{m}}^{k_{\alpha}}\delta_{2^{n}}^{\alpha})=\delta_2^j$, which means that $\mathrm{Col}_{k_{\alpha}}(\hat{K}_{i}\delta_{2^{n}}^{\alpha})=\delta_{2}^{j}$.
It follows that,  $\mathrm{sgn}(\hat{K}_{i}\delta_{2^{n}}^{\alpha}\mathbf{1}_{2^{m}})=\mathbf{1}_{2}$ holds for $\alpha\in\{1,\ldots,2^{n}\},i\in\{1,\ldots,p\}$ due to the fact that $(\delta_{2}^{l})^{\top}\widehat{\Xi}^{i}\delta_{2}^{1}\neq (\delta_{2}^{l})^{\top}\widehat{\Xi}^{i}\delta_{2}^{2}$.
For each state $x\in\Delta_{2^{n}}$ (assume $x=\delta_{2^{n}}^{\alpha}\in\Gamma_{il_i},i\in\{1,\ldots,p\}$), let $\mathrm{Col}_{\bar{\eta}}(Kx)=\delta_{2^{m}}^{k_{\alpha}}$ and $\delta_{2^{p}}^{\bar{\eta}}=\ltimes_{i=1}^{p}\delta_{2}^{\bar{\eta}_{i}}$.
Then for each index $i\in\{1,\ldots,p\}$, one has $\mathrm{sgn}(\sum\limits_{\delta_{2^{n}}^{\alpha}\in\Gamma_{il_i}}H_iL\delta_{2^{m}}^{k_{\alpha}}\delta_{2^{n}}^{\alpha})=\delta_{2}^{\eta_i}$ and $\mathrm{Col}_{k_{\alpha}}(\hat{K}_i\delta_{2^{n}}^{\alpha})=\delta_{2}^{\eta_i}$.
As a result, $$(\delta_{2^{m}}^{k_{\alpha}})^{\top}(\mathop{\prod\nolimits_{H}}\limits_{i=1\ \ }^{p\ \ }((\hat{K}_{i}\delta_{2^{n}}^{\alpha})^{\top}\delta_{2}^{\eta_{i}}))=1.$$
Since any set of $\eta_1,\ldots,\eta_p$ has a set of $\bar{\eta}_1,\ldots,\bar{\eta}_p$ corresponding to it, one concludes that $\mathop{\prod\nolimits_{H}}\limits_{i=1\ \ }^{p\ \ }((\hat{K}_{i}\delta_{2^{n}}^{\alpha})^{\top}\delta_{2}^{\eta_{i}})\neq\mathbf{0}_{2^{m}}$ holds for $\eta_i\in\{1,2\},i\in\{1,\ldots,p\}$.

(Sufficiency) If (\ref{e3.2.3}) holds, then the state feedback matrix $K$ can be designed as follows.
For each input $v(t)=\delta_{2^{p}}^{\eta}$ composed of $v_{i}(t)=\delta_{2}^{\eta_{i}},i\in\{1,\ldots,p\}$, take  \[\mathrm{Col}_{\eta}(Kx(t))=\delta_{2^{m}}^{k}\ \mathrm{if}\ \mathrm{Row}_{k}(\mathop{\prod\nolimits_{H}}\limits_{i=1\ \ }^{p\ \ }((\hat{K}_{i}x(t))^{\top}\delta_{2}^{\eta_{i}}))=1.\]
For each index $i\in\{1,\ldots,p\}$ and any given state $x=\delta_{2^{n}}^{\alpha}$, one has $K\delta_{2^{n}}^{\alpha}V(\hat{\Omega}_{ij})\circ(\mathrm{Row}_{j}(\hat{K}_{i}\delta_{2^{n}}^{\alpha}))^{\top}\neq\mathbf{0}_{2^{m}}.$
In view of (\ref{e3.2.4}), it follows that \[\mathrm{sgn}(\sum\limits_{\delta_{2^{n}}^{\alpha}\in\Gamma_{il}}H_{i}L(K\delta_{2^{n}}^{\alpha}V(\hat{\Omega}_{ij}))\delta_{2^{n}}^{\alpha})=\delta_{2}^{j},j,l=1,2,\]
which means $(\delta_{2}^{l})^{\top}\widehat{\Xi}^{i}\delta_{2}^{j}\in\Delta,j,l=1,2$.
Using the fact that $\mathrm{sgn}(\hat{K}_{i}\delta_{2^{n}}^{\alpha}\mathbf{1}_{2^{m}})=\mathbf{1}_{2},\alpha\in\{1,\ldots,2^{n}\}$, one has \[(\delta_{2}^{l})^{\top}\widehat{\Xi}^{i}\delta_{2}^{1}\neq(\delta_{2}^{l})^{\top}\widehat{\Xi}^{i}\delta_{2}^{2},l=1,2.\]
As a result, for each index $i\in\{1,\ldots,p\}$, the constructed Boolean matrix $\widehat{\Xi}^{i}$ satisfies (\ref{e3.3.4}).
BCN (\ref{e3.2.2}) is one-step transition IO-decoupled as per Theorem \ref{th1}, \ie, BCN (\ref{e2}) is one-step transition IO-decoupled under the state feedback control (\ref{e3.2.1}).
\end{proof}

A $2^{m}\times 2^{n+p}$ Boolean matrix $\hat{K}$ is constructed by
\begin{equation}\label{e3.3.2}
\mathrm{Col}_{\eta}(\hat{K}\delta_{2^{n}}^{\alpha})=\mathop{\prod\nolimits_{H}}\limits_{i=1\ \ }^{p\ \ }(\hat{K}_{i}\delta_{2^{n}}^{\alpha})^{\top}\delta_{2}^{\eta_{i}},
\end{equation}
where $\alpha\in\{1,\ldots,2^{n}\},\eta\in\{1,\ldots,2^{p}\}$ and $\delta_{2^{p}}^{\eta}=\ltimes_{i=1}^{p}\delta_{2}^{\eta_{i}}$. According to Theorem \ref{th3}, if $\mathrm{sgn}(\hat{K}_{i}\delta_{2^{n}}^{\alpha}\mathbf{1}_{2^{m}})=\mathbf{1}_{2}$ holds for   $\alpha\in\{1,\ldots,2^{n}\},i\in\{1,\ldots,p\}$, then the state feedback matrix $K$ exists if and only if \[\mathrm{Col}_{i}(\hat{K})\neq\mathbf{0}_{2^{m}}\]
holds for $i\in\{1,\ldots,2^{n+p}\}$.
The state feedback matrix $K$ can be designed as follow:
\[
\mathrm{Col}_{\eta}(K)\in\{\delta_{2^{m}}^{k}|[\hat{K}]_{k\eta}=1,k\in\{1,\ldots,2^{m}\}\},
\]
where $\eta\in\{1,\ldots,2^{n+p}\}$.

For each $\alpha\in\{1,\ldots,2^{n}\}$, one has $\mathrm{Col}_{j}(\hat{K}\delta_{2^{n}}^{\alpha})\circ\mathrm{Col}_{k}(\hat{K}\delta_{2^{n}}^{\alpha})=\mathbf{0}_{2^{m}},j\neq k$ from  $((\hat{K}_{i}\delta_{2^{n}}^{\alpha})^{\top}\delta_{2}^{1})\circ((\hat{K}_{i}\delta_{2^{n}}^{\alpha})^{\top}\delta_{2}^{2})=\mathbf{0}_{2^{m}},i\in\{1,\ldots,p\}$.
Thus, for each state $x\in\Delta_{2^{n}}$ (assume $x=\delta_{2^{n}}^{\alpha}$), $2^{p}$
disjoint combinations
$\{\delta_{2^{m}}^{k_{\alpha}^{j_1}},\ldots,\delta_{2^{m}}^{k_{\alpha}^{j_{l_j}}}\}$, guaranteeing that one column of $\hat{K}\delta_{2^{n}}^{\alpha}$ equals to $\sum\limits_{i=1}^{l_j}\delta_{2^{m}}^{k_{\alpha}^{i}}$, 
are needed to determine a part of the state feedback matrix $K$ (\ie, $K\delta_{2^{n}}^{\alpha}$).
Based on this, we have the following conclusions:
\begin{enumerate}
\item If $m<p$, then for each $\alpha\in\{1,\ldots,2^{n}\}$, the Boolean matrix $\hat{K}\delta_{2^{n}}^{\alpha}$ has at least one zero column, which implies that the state feedback matrix $K$ does not exist.
\item If $m=p$, then the state feedback matrix $K$ exists if and only if $\hat{K}\in\mathcal{L}_{2^{m}\times 2^{n+p}}$.
\end{enumerate}

To sum up, the size relation of $m$ and $p$ affects the existence of the state feedback.
A procedure (Algorithm \ref{alg1}) for the design of the state feedback matrix is given as follows:
\begin{algorithm}[!htbp]
\renewcommand{\algorithmicrequire}{\textbf{Input:}}
\renewcommand{\algorithmicensure}{\textbf{Output:}}
\caption{The design of the state feedback matrix}
\label{alg1}
\begin{algorithmic}[1]
\Require $L$, $H_{i}$, $\Gamma_{ij},i\in\{1,\ldots,p\},j=1,2$.
\Ensure $K$.
\If{$m\geq p$,}
\For {$i=1$ to $p$,}
\State{let $\hat{K}_{i}=\mathbf{0}_{2\times 2^{n+m}}$;}
\If {$\mathrm{sgn}(\sum\limits_{\delta_{2^{n}}^{\alpha}\in\Gamma_{il}}H_{i}L\delta_{2^{m}}^{k_{\alpha}}\delta_{2^{n}}^{\alpha})=\delta_{2}^{j}$,}
\State{$\mathrm{Col}_{k_{\alpha}}(\hat{K}_{i}\delta_{2^{n}}^{\alpha})=\delta_{2}^{j}$;}
\EndIf
\For{$\alpha=1$ to $2^{n}$,}
\If{$\mathrm{sgn}(\hat{K}_{i}\delta_{2^{n}}^{\alpha}\mathbf{1}_{2^{m}})\neq\mathbf{1}_{2}$,}
\State{\textbf{stop};}
\EndIf
\EndFor
\EndFor
\State{compute $\hat{K}$ according to (\ref{e3.3.2});}
\If{$\mathrm{Col}_{i}(\hat{K})\neq\mathbf{0}_{2^{m}},i\in\{1,\ldots,2^{n+p}\}$,}
\If {$m=p$,}
\State{$K=\hat{K}$;}
\Else
\For{$i=1$ to $2^{n+p}$,}
\State{find a $k$ such that $[\hat{K}]_{ki}=1$;}
\State{$\mathrm{Col}_{i}(K)=\delta_{2^{m}}^k$.}
\EndFor
\EndIf
\EndIf
\EndIf
\end{algorithmic}
\end{algorithm}

\begin{remark}\rm
    The computational complexity of Algorithm \ref{alg1} is $O(4^{n+p})$.
\end{remark}

\section{Comparative analysis}\label{sec4}
In this section, we provide some comparisons with the existing results in the literature.

\subsection{Comparison with the one-step transition IO-decoupling  discussed in \cite{ValcherME2017}}

The definition of the one-step transition IO-decoupling in \cite{ValcherME2017} is given below.

\begin{definition}\label{def1}\rm\cite{ValcherME2017}
BCN (\ref{e2}) with inputs and outputs having the same cardinality (\ie, $m=p$) is said to be one-step  transition IO-decoupled if for every index $i\in\{1,\ldots,m\}$, every pair of states $x,\hat{x}\in\Delta_{2^{n}}$ and every pair of inputs $u,\hat{u}\in\Delta_{2^{m}}$ satisfying
\[
\left\{
\begin{array}{l}
H_{i}x=H_{i}\hat{x},\\
u_{i}=\hat{u}_{i},\\
\end{array}
\right.
\]
ensure that
\[H_{i}Lux=H_{i}L\hat{u}\hat{x}.\]
\end{definition}

According to Definition \ref{def1}, if BCN (\ref{e2}) is one-step transition IO-decoupled, then for each index $i\in\{1,\ldots,m\}$ and each state $x\in\Delta_{2^{n}}$, each pair of inputs $u,\hat{u}\in\Omega_{ik},k=1,2$ ensures
$H_{i}Lux=H_{i}L\hat{u}x$,
\ie,
$\mathrm{sgn}(H_{i}LV(\Omega_{ik})x)\in\Delta$.
Since any two states $x,\hat{x}\in\Gamma_{il},l=1,2$ satisfy
$\mathrm{sgn}(H_{i}LV(\Omega_{ik})x)=\mathrm{sgn}(H_{i}LV(\Omega_{ik})\hat{x})$, it follows that
$\mathrm{sgn}(\sum\limits_{x\in\Gamma_{il}}H_{i}LV(\Omega_{ik})x)$
$\in\Delta, k,l=1,2$.
Therefore, the constructed Boolean matrix $\Xi^{i}$ can be used to verify whether a given BCN is one-step transition IO-decoupled as per Definition \ref{def1}.

\begin{proposition}\label{pro4.1}\rm
BCN (\ref{e2}), with inputs and outputs having the same dimension (\ie, $m=p$), is one-step transition IO-decoupled, as per Definition \ref{def1}, if and only if for each index $i\in\{1,\ldots,m\}$, the constructed Boolean matrix $\Xi^{i}$ satisfies
\begin{equation}\label{eq_def_2}
(\delta_{2}^{l})^{\top}\Xi^{i}\in\mathcal{L}_{2\times 2},
\end{equation}
where $l=1,2$.
\end{proposition}

\begin{proof}
BCN (\ref{e2}) is one-step transition IO-decoupled as per Definition \ref{def1} if and only if for each index $i\in\{1,\ldots,m\}$, one has
\[
\mathrm{sgn}(\sum\limits_{\delta_{2^{n}}^{\alpha}\in\Gamma_{ij}}H_{i}LV(\Omega_{ik})\delta_{2^{n}}^{\alpha})\in\Delta,k,l=1,2,
\]
which is equivalent to
\[
(\delta_{2}^{l})^{\top}\Xi^{i}\delta_{2}^{k}\in\Delta,k,l=1,2,
\]
from which \eqref{eq_def_2} follows.
\end{proof}

\begin{remark}\label{re4.1}\rm
In view of Theorem \ref{th1} and Proposition \ref{pro4.1}, it is clear that the verification condition of the one-step transition IO-decoupling, as per Definition \ref{def2}, ensures that $(\delta_{2}^{l})^{\top}\Xi^{i}\delta_{2}^{1}\neq(\delta_{2}^{l})^{\top}\Xi^{i}\delta_{2}^{2},l=1,2$, which is not required for the one-step transition IO-decoupling as per Definition \ref{def1}. Note that $(\delta_{2}^{l})^{\top}\Xi^{i}\delta_{2}^{1}\neq(\delta_{2}^{l})^{\top}\Xi^{i}\delta_{2}^{2}$, ensures that every output $y_i$ is controlled by a unique input $u_i$.
\end{remark}
If BCN (\ref{e2}) is one-step transition IO-decoupled as per Definition \ref{def1}, then an explicit IO-mapping can be expressed as follows:
\[
y_{i}(t+1)=[(\delta_{2}^{1})^{\top}\Xi^{i}\ (\delta_{2}^{2})^{\top}\Xi^{i}]y_{i}(t)u_{i}(t),
\]
where $i\in\{1,\ldots,m\}$ and $(\delta_2^l)^{\top}\Xi^i\in\mathcal{L}_{2\times 2}$ holds for $l=1,2$.\\
Letting $z_{i}(t)=y_{i}(t),i\in\{1,\ldots,m\}$, one has
\[
\left\{
\begin{array}{l}
z_{i}(t+1)=[(\delta_{2}^{1})^{\top}\Xi^{i}\ (\delta_{2}^{2})^{\top}\Xi^{i}]W_{[2,2]}u_{i}(t)z_{i}(t),\\
y_{i}(t)=z_{i}(t),\\
\end{array}
\right.
\]
which is referred to as the one-step IO-decoupled canonical form in the case where the conditions in Proposition \ref{pro4.1} are satisfied.
Furthermore, the design of state feedback for one-step transition IO-decoupling was not addressed in \cite{ValcherME2017}. In what follows, we will address the one-step transition IO-decoupling as per Definition \ref{def1}.
According to Proposition \ref{pro4.1}, BCN (\ref{e3.2.2}) is one-step transition IO-decoupled as per Definition \ref{def1} if and only if for each index $i\in\{1,\ldots,m\}$, the Boolean matrix $\widehat{\Xi}^{i}$ given in Subsection \ref{subsec3.2} satisfies
\[
(\delta_{2}^{l})^{\top}\widehat{\Xi}^{i}\in\mathcal{L}_{2\times2},l=1,2.
\]
This means that the designed state feedback matrix $K$ ensures that
$(\delta_{2}^{l})^{\top}\widehat{\Xi}^{i}\delta_{2}^{k}\in\Delta,k,l=1,2$.
Since $\mathrm{sgn}(\hat{K}_{i}\delta_{2^{n}}^{\alpha}\mathbf{1}_{2^{m}})=\mathbf{1}_{2},\alpha\in\{1,\ldots,2^{n}\}$ ensures $(\delta_{2}^{l})^{\top}\widehat{\Xi}^{i}\delta_{2}^{1}\neq(\delta_{2}^{l})^{\top}\widehat{\Xi}^{i}\delta_{2}^{2},l=1,2$,
the existence condition of one-step transition IO-decoupling state feedback (as per Definition \ref{def1}) does not require  $\mathrm{sgn}(\hat{K}_{i}\delta_{2^{n}}^{\alpha}\mathbf{1}_{2^{m}})=\mathbf{1}_{2},\alpha\in\{1,\ldots,2^{n}\}$.
However, if there exist $i\in\{1,\ldots,p\}$  and $\alpha\in\{1,\ldots,2^{n}\}$ such that $\mathrm{sgn}(\hat{K}_{i}\delta_{2^{n}}^{\alpha}\mathbf{1}_{2^{m}})\neq\mathbf{1}_{2}$, then, without loss of generality assuming that $\mathrm{sgn}(\hat{K}_{i}\delta_{2^{n}}^{\alpha}\mathbf{1}_{2^{m}})=\delta_{2}^{1}$, one has
\[
\begin{array}{l}
(\mathop{\prod\nolimits_{H}}\limits_{j=1\ \ }^{i-1\ \ }((\hat{K}_{j}\delta_{2^{n}}^{\alpha})^{\top}\delta_{2}^{\eta_{j}}))\circ((\hat{K}_{i}\delta_{2^{n}}^{\alpha})^{\top}\delta_{2}^{2})\\
\ \ \ \ \ \ \ \ \ \ \ \ \ \ \ \ \ \ \ \circ(\mathop{\prod\nolimits_{H}}\limits_{j=i+1\ \ }^{p\ \ }((\hat{K}_{j}\delta_{2^{n}}^{\alpha})^{\top}\delta_{2}^{\eta_{j}}))=\mathbf{0}_{2^{m}},
\end{array}
\]
showing that the determination of the state feedback matrix $K$ is not possible.
In this case, a modified auxiliary matrix $\bar{K}_{i},i\in\{1,\ldots,p\}$ is proposed as follows:
\begin{equation}\label{e4.1}
\bar{K}_{i}\delta_{2^{n}}^{\alpha}=\left\{
\begin{array}{ll}
\hat{K}_{i}\delta_{2^{n}}^{\alpha},&\mathrm{if}\ \mathrm{sgn}(\hat{K}_{i}\delta_{2^{n}}^{\alpha}\mathbf{1}_{2^{m}})=\mathbf{1}_{2},\\
\mathbf{1}_{2^{m}}^{\top}\otimes\mathbf{1}_{2},&\mathrm{if}\ \mathrm{sgn}(\hat{K}_{i}\delta_{2^{n}}^{\alpha}\mathbf{1}_{2^{m}})\neq\mathbf{1}_{2},\\
\end{array}
\right.
\end{equation}
where $\alpha\in\{1,\ldots,2^{n}\}$ and  $\hat{K}_{i}$ is computed according to (\ref{e3.2.4}).
For each index $i\in\{1,\ldots,p\}$ and each state $x(t)\in\Delta_{2^{n}}$, we design the state feedback matrix $K$ such that $y_{i}(t+1)=H_{i}L(Kx(t)v(t))x(t)=v_{i}(t)$ if $\mathrm{sgn}(\hat{K}x(t)\mathbf{1}_{2^{m}})=\mathbf{1}_{2}$, otherwise, we force $K$ such that  $y_{i}(t+1)=\mathrm{sgn}(\hat{K}x(t)\mathbf{1}_{2^{m}})$ regardless of what value $v_{i}(t)$ takes.
Based on this, a new Boolean matrix $\bar{K}\in\mathcal{B}_{2^{m}\times 2^{n+p}}$ is constructed as follows:
\begin{equation}\label{eq4.2}
\mathrm{Col}_{\eta}(\bar{K}\delta_{2^{n}}^{\alpha})=\mathop{\prod\nolimits_{H}}\limits_{i=1\ \ }^{p\ \ }(\bar{K}_{i}\delta_{2^{n}}^{\alpha})^{\top}\delta_{2}^{\eta_{i}},
\end{equation}
which is used to ensure the existence of the state feedback matrix $K$.

\begin{proposition}\label{pro4.2}\rm
BCN (\ref{e2}) is one-step transition IO-decoupled (as per Definition \ref{def1}) under the state feedback control (\ref{e3.2.1}) if
\[
\mathrm{Col}_{i}(\bar{K})\neq\mathbf{0}_{2^{m}}
\]
holds for $i\in\{1,\ldots,2^{n+p}\}$.
\end{proposition}

\begin{proof}
The proof is omitted as it is similar to the proof of Theorem \ref{th3}.
\end{proof}

For any $\bar{K}_i\delta_{2^{n}}^{\alpha}$ satisfying $\bar{K}_i\delta_{2^{n}}^{\alpha}\neq\mathbf{1}^{\top}_{2^{m}}\otimes\mathbf{1}_2$, assume the $k$-th column of those $\bar{K}_i\delta_{2^{n}}^{\alpha}$ are equivalent and denoted by $\delta_{2}^{j}$.
Then one has $\mathrm{sgn}(\sum\limits_{\delta_{2^{n}}^{\alpha}\in\Gamma_{ij}}H_iL\delta_{2^{m}}^{k}\delta_{2^{n}}^{\alpha})=\delta_2^j$.
For any $\bar{K}_i\delta_{2^{n}}^{\alpha}$ satisfying $\bar{K}_i\delta_{2^{n}}^{k}=\mathbf{1}^{\top}_{2^{m}}\otimes\mathbf{1}_2$,
one has $\mathrm{sgn}(\hat{K}_i\delta_{2^{n}}^{\alpha}\mathbf{1}_{2^{m}})\neq\mathbf{1}_2$, without loss of generality, assume $\mathrm{sgn}(\hat{K}_i\delta_{2^{n}}^{\alpha}\mathbf{1}_{2^{m}})=\delta_2^1$.
It follows that $\mathrm{sgn}(\sum\limits_{\delta_{2^{n}}^{\alpha}\in\Gamma_{ij}}H_iL\delta_{2^{m}}^{k}\delta_{2^{n}}^{\alpha})=\delta_2^1$.
As a result, if $K=\delta_{2^{m}}^{k}\otimes\mathbf{1}_{2^{n+p}}^{\top}$, then
the given BCN is one-step transition IO-decoupled (as per Definition \ref{def1}) under the state feedback $u(t)=(\delta_{2^{m}}^{k}\otimes\mathbf{1}_{2^{n+p}}^{\top})x(t)v(t)$.
Unfortunately, such a state feedback matrix $K$ cannot be obtained from $\bar{K}$ constructed via (\ref{eq4.2}) due to the existence of
$\bar{K}_{i}\delta_{2^{n}}^{\alpha}\neq\mathbf{1}^{\top}_{2^{m}}\otimes\mathbf{1}_{2}$ leading to $\mathrm{Col}_{\eta}(\bar{K})\neq\mathrm{Col}_{\bar{\eta}}(\bar{K})$, where $\delta_{2^{p}}^{\eta}=\ltimes_{j=1}^{i-1}\delta_{2}^{\eta_j}\ltimes\delta_{2}^{1}\ltimes(\ltimes_{j=i+1}^{p}\delta_{2}^{\eta_j})$ and $\delta_{2^{p}}^{\bar{\eta}}=\ltimes_{j=1}^{i-1}\delta_{2}^{\eta_j}\ltimes\delta_{2}^{2}\ltimes(\ltimes_{j=i+1}^{p}\delta_{2}^{\eta_j})$.
Consequently, $\mathrm{Col}_i(\bar{K})\neq\mathbf{0}_{2^{m}}$ is just a sufficient condition.
For instance, given a BCN with $m=2$ and $p=2$, if for each $\alpha\in\{1,\ldots,2^{n}\}$, the obtained $\bar{K}_1\delta_{2^{n}}^{\alpha},\bar{K}_2\delta_{2^{n}}^{\alpha}$ are $\bar{K}_1\delta_{2^{n}}^{\alpha}=\delta_{2}[1\ 2\ 2\ 2]$ and $\bar{K}_2\delta_{2^{n}}^{\alpha}=\delta_{2}[1\ 1\ 1\ 2]$, then one has
\[
\bar{K}\delta_{2^{n}}^{\alpha}=\left[
\begin{array}{cccc}
1&0&0&0\\
0&0&1&0\\
0&0&1&0\\
0&0&0&1\\
\end{array}
\right]
\]
where $\alpha\in\{1,\ldots,2^{n}\}$.
Due to the fact that $\mathrm{Col}_2(\bar{K})=\mathbf{0}_{4}$, we cannot determine whether a state feedback exists such that the given BCN is one-step transition IO-decoupled as per Definition \ref{def1}.
However, if one takes $K=\delta_{4}^{1}\otimes\mathbf{1}^{\top}_{2^{n+2}}$, then for $i\in\{1,2\}$, one has $\mathrm{sgn}(\sum\limits_{\delta_{2^{n}}^{\alpha}\in\Gamma_{ij}}H_iL\delta_{4}^{1}\delta_{2^{n}}^{\alpha})=\delta_2^1$, which means that $(\delta_2^j)^{\top}\widehat{\Xi}^{i}\in\mathcal{L}_{2\times 2},i\in\{1,2\}$.
Hence, the given BCN is one-step transition IO-decoupled (as per Definition \ref{def1}) under the state feedback $u(t)=(\delta_{4}^{1}\otimes\mathbf{1}^{\top}_{2^{n+2}})x(t)v(t)$.

The procedure for the design of a state feedback, leading to a one-step transition IO-decoupled BCN (as per Definition \ref{def1}), is given in
Algorithm \ref{alg1} after altering Steps 8-10 as follows:
\begin{algorithm}[!htbp]
\begin{algorithmic}
\If{$\mathrm{sgn}(\hat{K}_{i}\delta_{2^{n}}^{\alpha}\mathbf{1}_{2^{m}})\neq\mathbf{1}_{2}$,}
\State{$\hat{K}_{i}\delta_{2^{n}}^{\alpha}=\mathbf{1}_{2^{m}}^{\top}\otimes\mathbf{1}_{2}$;}
\EndIf
\end{algorithmic}
\end{algorithm}

\subsection{Comparison with the general IO-decoupling of \cite{ValcherME2017}}
Reference \cite{ValcherME2017} proposed a more general definition of BCN IO-decoupling, which is given below.

\begin{definition}\label{def3}\rm\cite{ValcherME2017}
BCN (\ref{e2}) with inputs and outputs having the same cardinality (\ie, $m=p$) is said to be IO-decoupled if for every index $i\in\{1,\ldots,m\}$ and every initial state $x(0)\in\Delta_{2^{n}}$, if $\{u(t)\}_{t=0}^{+\infty}$ and $\{\hat{u}(t)\}_{t=0}^{+\infty}$ are two input sequences characterized by the fact that their $i$-th entries coincide at every time instant, \ie,
\[
u_{i}(t)=\hat{u}_{i}(t),\forall t\in\mathbb{N},
\]
then the output sequences $\{y(t)\}_{t=0}^{+\infty}$ and $\{\hat{y}(t)\}_{t=0}^{+\infty}$, generated by BCN (\ref{e2}) corresponding to $x(0),\{u(t)\}_{t=0}^{+\infty}$ and $\{\hat{u}(t)\}_{t=0}^{+\infty}$, respectively, satisfy
\[
y_{i}(t)=\hat{y}_{i}(t),\forall t\in\mathbb{N}.
\]
\end{definition}

In light of the results obtained in \cite{ValcherME2017}, the authors in \cite{Liyifeng2022a} discussed the IO-decoupling of BCNs, and presented a necessary and sufficient vertex partitioning condition. However, there is no result on the state feedback control design for the IO-decoupling (as per Definition \ref{def3}). Using the proposed approach, a sufficient condition can be provided for the IO-decoupling control design (as per Definition \ref{def3}) in the following proposition.

\begin{proposition}\label{pro4.3}\rm
If $\mathrm{Col}_{i}(\bar{K})\neq\mathbf{0}_{2^{m}}$
holds for $i\in\{1,\ldots,2^{n+p}\}$, then
BCN (\ref{e2}) is IO-decoupled (as per Definition \ref{def3}) under the state feedback control (\ref{e3.2.1}).
Moreover, each column of $K$ can be designed as follows:
\[\mathrm{Col}_{i}(K)\in\{\delta_{2^{m}}^{k}|[\bar{K}]_{ki}=1,k\in\{1,\ldots,2^{m}\}\},\]
where $i\in\{1,\ldots,2^{n+p}\}$.
\end{proposition}

\begin{proof}
The proof follows directly from the fact that the one-step transition IO-decoupling in the sense of Definition \ref{def1} is a special case of the IO-decoupling in the sense of Definition \ref{def3}.
\end{proof}

By virtue of Proposition \ref{pro4.2}, $\mathrm{Col}_{\eta}(\bar{K})\neq\mathbf{0}_{2^{m}},\eta\in\{1,\ldots,2^{n+p}\}$ results in such a situation that for each index $i\in\{1,\ldots,p\}$, the output sequence $\{y_i(t)\}_{t=0}^{+\infty}$ generated by BCN (\ref{e3.2.2}) remains unchanged for any initial state $x(0)\in\Gamma_{ij_{0}},j_{0}\in\{1,2\}$ and any input sequence $\{v(t)\}_{t=0}^{+\infty}$ with $v(t)\in\hat{\Omega}_{ik_t},k_t\in\{1,2\},t\in\mathbb{Z}_{+}$, where $\Gamma_{ij}$ and $\hat{\Omega}_{ik}$ are defined in subsection 3.1 and subsection 3.3 respectively.
However, Definition \ref{def3} only requires that for each index $i\in\{1,\ldots,p\}$, the output sequence $\{y_i(t)\}_{t=0}^{+\infty}$ generated by BCN (\ref{e3.2.2}) corresponding to each initial state $x(0)\in\Delta_{2^{n}}$ remains unchanged for any input sequence $\{v(t)\}_{t=0}^{+\infty}$ with $v(t)\in\hat{\Omega}_{ik_t},k_t\in\{1,2\},t\in\mathbb{Z}_{+}$.
Therefore, the results in Proposition \ref{pro4.3} are limited.
It is worth noting that for each index $i\in\{1,\ldots,p\}$, $\widehat{\Xi}^{i}$ satisfying $[(\delta_{2}^{1})^{\top}\widehat{\Xi}^{i}\ (\delta_{2}^{2})^{\top}\widehat{\Xi}^{i}]\in\mathcal{L}_{2\times 4}$ means that the output $y_i(t+1)$ generated by BCN (\ref{e3.2.2}) remains unchanged for any state $x(t)\in\Gamma_{ij_t},j_t\in\{1,2\}$ and any input $v(t)\in\hat{\Omega}_{ik_t},k_t\in\{1,2\}$.
In our future research investigations, we will work on the construction of some matrices, similar to $\widehat{\Xi}^{1},\ldots,\widehat{\Xi}^{p}$, that have to satisfy some conditions to ascertain whether a given BCN is IO-decoupled as per Definition \ref{def3} under a state feedback.
Specifically, for each index $i\in\{1,\ldots,p\}$ and any two states $x,\bar{x}\in\Delta_{2^{n}}$, if $x$ and $\bar{x}$ have the same output sequence $\{y_i(t)\}_{t=0}^{+\infty}$ under any input sequence $\{v(t)\}_{t=0}^{+\infty}$ with $v(t)\in\hat{\Omega}_{ik_t},k_t\in\{1,2\},t\in\mathbb{Z}_{+}$, then we put $x$ and $\bar{x}$ into the same state subset.
Repeating this process, we can obtain a set of state subsets, denoted by $\Theta_{i1},\ldots,\Theta_{i\mu_i}$, where $i\in\{1,\ldots,p\}$.
Based on this, a set of Boolean matrices $\widetilde{\Xi}^{1},\ldots,\widetilde{\Xi}^{p}$ are constructed as follows:
\begin{enumerate}
\item for each $\alpha\in\{1,\ldots,2^{n}\}$, let \[L'\delta_{2^{n}}^{\alpha}=LW_{[2^{n},2^{m}]}\delta_{2^{n}}^{\alpha}[V(\hat{\Omega}_{i1})\ V(\hat{\Omega}_{i2})].\]
\item for each index $i\in\{1,\ldots,p\}$, $\widetilde{\Xi}^{i}$ is obtained according to

\[\mathrm{Row}_{k}(\widetilde{\Xi}^{i})=\mathrm{sgn}(\sum\limits_{\delta_{2^{n}}^{j}\in\Theta_{ik}}\mathrm{Row}_{j}(L')).\]
\end{enumerate}
In view of the number of ``ones'' in each column of the constructed Boolean matrix, we determine whether a given BCN is IO-decoupling as per Definition \ref{def3} under a state feedback.

\subsection{Comparison with the IO-decomposition approach}
References \cite{Fushihua2017,Panjinfeng2018} defined the IO-decoupling of BCNs using the IO-decomposition as follows:

\begin{definition}\label{def4}\rm\cite{Fushihua2017,Panjinfeng2018}
The IO-decoupling of BCN (\ref{e2}) is solvable if there exists a coordinate transformation $z=Px$ and a state feedback control (\ref{e3.2.1}) such that BCN (\ref{e2}) can be converted into
\begin{equation}\label{e4.2}
\left\{
\begin{array}{l}
z^{i}(t+1)=F_{i}v_{i}(t)z^{i}(t),\\
z^{p+1}(t+1)=F_{p+1}v(t)z(t),\\
y_{i}(t+1)=E_{i}z^{i}(t),\\
\end{array}
\right.
\end{equation}
where $z=\ltimes_{i=1}^{p+1}z^{i}\in\Delta_{2^{n}},z^{i}\in\Delta_{2^{n_{i}}}, F_{i}\in\mathcal{L}_{2^{n_{i}}\times 2^{n_{i}+1}},$
$E_{i}\in\mathcal{L}_{2\times 2^{n_{i}}},i\in\{1,\ldots,p\},F_{p+1}\in\mathcal{L}_{2^{n_{p+1}}\times2^{n+p}}$ and $n=n_{1}+\cdots+n_{p+1}$.
\end{definition}

Using the proposed approach, a sufficient condition for the IO-decoupling in the sense of Definition \ref{def4} is presented.

\begin{proposition}\label{pro4.4}\rm
If \[
\left\{
\begin{array}{l}
\mathrm{Col}_{i}(\bar{K})\neq\mathbf{0}_{2^{m}},i\in\{1,\ldots,2^{n+p}\},\\
(H_{1}\ast\cdots\ast H_{p})\mathbf{1}_{2^{n}}=2^{n-p}\mathbf{1}_{2^{p}},\\
\end{array}
\right.
\]
Then the following statements hold:
\begin{enumerate}
\item Taking the columns of the state feedback matrix $K$ as follows
\[
\mathrm{Col}_{i}(K)\in\{\delta_{2^{m}}^{k}|[\bar{K}]_{ki}=1,k\in\{1,\ldots,2^{m}\}\}
\]
leads to an IO-decoupled closed loop system as per Definition \ref{def4}.
\item There exists a $P_{p+1}\in\mathcal{L}_{2^{n-p}\times 2^{n}}$ such that $P=H_1\ast\cdots\ast H_p\ast P_{p+1}$ is nonsingular, and the IO-decomposed form is given by
\begin{equation}\label{eq4.1}
\left\{
\begin{array}{l}
z^{i}(t+1)=[(\delta_{2}^{1})^{\top}\widehat{\Xi}^{i}\ (\delta_{2}^{2})^{\top}\widehat{\Xi}^{i}]W_{[2,2]}v_{i}(t)z^{i}(t),\\
z^{p+1}(t+1)=P_{p+1}\widehat{L}(I_{2^{p}}\otimes P^{\top})v(t)z(t),\\
y_{i}(t)=z^{i}(t),\\
\end{array}
\right.
\end{equation}
where $i\in\{1,\ldots,p\}$ and $\widehat{L}=LKW_{[2^{p},2^{n}]}(I_{2^{p}}\otimes\Phi_{2^{n}}).\Phi_{2^{n}}=\mathrm{diag}\{\delta_{2^{n}}^{1},\ldots,\delta_{2^{n}}^{2^{n}}\}$.
\end{enumerate}
\end{proposition}

\begin{proof}
Let us prove item (1).
From Proposition \ref{pro4.2}, if $\mathrm{Col}_{i}(\bar{K})\neq\mathbf{0}_{2^{m}},i\in\{1,\ldots,2^{n+p}\}$, then BCN (\ref{e2}) is one-step transition IO-decoupled as per Definition \ref{def1} under the state feedback control (\ref{e3.2.1}).
This means that taking the columns of the state feedback matrix $K$ as follows
\[
\mathrm{Col}_{i}(K)\in\{\delta_{2^{m}}^{k}|[\bar{K}]_{ki}=1,k\in\{1,\ldots,2^{m}\}\}
\]
leads to a one-step transition IO-decoupled closed-loop system as per Definition \ref{def1}.
By virtue of Theorem \ref{th5}, the IO-decomposed form (\ref{e4.2}) of the resulting closed-loop system exists due to the fact that $(H_1\ast\cdots\ast H_p)\mathbf{1}_{2^{n}}=2^{n-p}\mathbf{1}_{2^{p}}$.
Therefore, the resulting closed-loop system is IO-decoupled as per Definition \ref{def4}.

Now, we prove item (2). Letting $\Lambda_{j}=\{x|(H_1\ast\cdots\ast  H_p)x=\delta_{2^{p}}^{j}\},j\in\{1,\ldots, 2^{p}\}$, one has $|\Lambda_{j}|=2^{n-p},j\in\{1,\ldots,2^{p}\}$ due to the fact that $(H_{1}\ast\cdots\ast H_{p})\mathbf{1}_{2^{n}}=2^{n-p}\mathbf{1}_{2^{p}}$.
For $j=1,\ldots,2^{p}$, denote $\Lambda_{j}=\{\delta_{2^{n}}^{\lambda_{j1}},\ldots,\delta_{2^{n}}^{\lambda_{j|\Lambda_{j}|}}\}$. Define a one-to-one mapping $\psi$ from $\Delta_{2^{n}}$ to $\Delta_{2^{n}}$ such that $\psi(\delta_{2^{n}}^{\lambda_{jk}})=\delta_{2^{p}}^{j}\ltimes\delta_{2^{n-p}}^{k}$,
where $\delta_{2^{n}}^{\lambda_{jk}}\in\Lambda_j,j\in\{1,\ldots,2^{p}\}$ and $k\in\{1,\ldots,2^{n-p}\}$.
Then two surjections are induced from $\psi$:
\begin{enumerate}
\item $\psi_1$ is a surjection from $\Delta_{2^{n}}$ to $\Delta_{2^{p}}$ such that $\psi_1(\delta_{2^{n}}^{\lambda_{jk}})=\delta_{2^{p}}^{j}$, where $\delta_{2^{n}}^{\lambda_{jk}}\in\Lambda_{j},j\in\{1,\ldots,2^{p}\}$ and $k\in\{1,\ldots,2^{n-p}\}$.
\item $\psi_2$ is a surjection from $\Delta_{2^{n}}$ to $\Delta_{2^{n-p}}$ such that $\psi_2(\delta_{2^{n}}^{\lambda_{jk}})=\delta_{2^{n-p}}^{k}$, where $\delta_{2^{n}}^{\lambda_{jk}}\in\Lambda_{j},j\in\{1,\ldots,2^{p}\}$ and $k\in\{1,\ldots,2^{n-p}\}$.
\end{enumerate}
These three mappings can be expressed as $\psi(x)=Px,\psi_1(x)=(H_1\ast\cdots\ast H_p)x$ and $\psi_2(x)=P_{p+1}x$, where $P\in\mathcal{L}_{2^{n}\times 2^{n}}$ is nonsingular, $P_{p+1}\in\mathcal{L}_{2^{n-p}\times 2^{n}}$.
For each $j\in\{1,\ldots,2^{p}\}$ and $k\in\{1,\ldots,2^{n-p}\}$, one has
\[
\begin{array}{rcl}
\psi(\delta_{2^{n}}^{\lambda_{jk}})&=&\delta_{2^{p}}^{j}\ltimes\delta_{2^{n-p}}^{k}\\
\ &=&\psi_{1}(\delta_{2^{n}}^{\lambda_{jk}})\ltimes\psi_{2}(\delta_{2^{n}}^{\lambda_{jk}})\\
\ &=&((H_{1}\ast\cdots\ast H_{p})\delta_{2^{n}}^{\lambda_{jk}})\ltimes(P_{p+1}\delta_{2^{n}}^{\lambda_{jk}})\\
\ &=&(H_{1}\ast\cdots\ast H_{p}\ast P_{p+1})\delta_{2^{n}}^{\lambda_{jk}}.\\
\end{array}
\]
Due to the arbitrariness of $j$ and $k$,
we determine that $P=H_1\ast\cdots\ast H_p\ast P_{p+1}$ is nonsingular.
Note that for each index $i\in\{1,\ldots,p\}$, the Boolean matrix $\widehat{\Xi}^{i}$ can be determined, and an explicit mapping between $y_{i}(t+1)$ and $v_{i}(t),y_{i}(t)$ can be given via
\[
y_{i}(t+1)=[(\delta_{2}^{1})^{\top}\widehat{\Xi}^{i}\ (\delta_{2}^{2})^{\top}\widehat{\Xi}^{i}]y_{i}(t)v_{i}(t).
\]
For each index $i\in\{1,\ldots,p\}$, letting $z^{i}(t)=y_i(t)$, one has
\[
\left\{
\begin{array}{l}
z^{i}(t+1)=[(\delta_{2}^{1})^{\top}\widehat{\Xi}^{i}\ (\delta_{2}^{2})^{\top}\widehat{\Xi}^{i}]W_{[2,2]}v_{i}(t)z^{i}(t),\\
y_{i}(t)=z^{i}(t).\\
\end{array}
\right.
\]
Letting $z^{p+1}(t)=P_{p+1}x(t)$, one has $z^{p+1}(t+1)=P_{p+1}\widehat{L}(I_{2^{p}}\otimes P^{\top})v(t)z(t)$, leading to the IO-decomposed form (\ref{eq4.1}).
\end{proof}

The state feedback control design approaches proposed in \cite{Fushihua2017,Panjinfeng2018} rely on the IO-decomposed form (\ref{e4.2}) which may not exist for a certain class of systems. In contrast, our proposed approach does not rely on the existence of the IO-decomposed form, but if such a form exists, it is obtained in a straightforward manner as stated in Proposition \ref{pro4.4}.

\section{Illustrative examples}\label{sec5}
In this section, two examples are provided to show the effectiveness of the proposed approach.

\begin{example}\label{ex1}\rm
Consider the following BCN given in \cite{Liyifeng2022a}, with 2 input nodes, 2 output nodes and 3 state nodes:
\begin{equation}\label{e4}
\left\{
\begin{array}{l}
x(t+1)=Lu(t)x(t),\\
y_{1}(t)=H_{1}x(t),\\
y_{2}(t)=H_{2}x(t),\\
\end{array}
\right.
\end{equation}
where $L=\delta_{8}[3\ 4\ 1\ 2\ 3\ 4\ 3\ 4\ 2\ 1\ 4\ 3\ 4\ 3\ 4\ 3\ 5\ 6\ 5\ 6\ 7\ 8\ 5\ 6\ 6\ 5\ 6$
$5\ 8\ 7\ 6\ 5]$, $H_{1}=\delta_{2}[1\ 1\ 1\ 1\ 2\ 2\ 2\ 2]$ and $H_{2}=\delta_{2}[1\ 2\ 1\ 2\ 1\ 2$
$1\ 2]$.
According to (\ref{e3}), $\Xi^{1}$ and $\Xi^2$ are computed as follows:
\[
\Xi^{1}=\left[
\begin{array}{ll}
\delta_{2}^{1}&\delta_{2}^{2}\\
\delta_{2}^{1}&\delta_{2}^{2}\\
\end{array}
\right]_{,}\
\Xi^{2}=\left[
\begin{array}{ll}
\delta_{2}^{1}&\delta_{2}^{2}\\
\delta_{2}^{2}&\delta_{2}^{1}\\
\end{array}
\right]_{.}\
\]
Due to the fact that $(\delta_2^l)^{\top}\Xi^i\mathbf{1}_{2}=\mathbf{1}_2,i,l=1,2$, BCN (\ref{e4}) is one-step transition IO-decoupled as per Theorem \ref{th1}.
The mapping between $y_{i}(t+1)$ and $u_{i}(t),y_{i}(t)$ can be obtained as follows:
\[
\begin{array}{l}
y_{1}(t+1)=\delta_{2}[1\ 2\ 1\ 2]y_{1}(t)u_{1}(t),\\
y_{2}(t+1)=\delta_{2}[1\ 2\ 2\ 1]y_{2}(t)u_{2}(t).\\
\end{array}\]
Let $z_{i}(t)=y_{i}(t),i=1,2$. Then the one-step transition IO-decoupled canonical form of BCN \eqref{e4} is as follows:
\[
\left\{
\begin{array}{l}
z_{1}(t+1)=u_{1}(t),\\
z_{2}(t+1)=\delta_{2}[1\ 2\ 2\ 1]u_{2}(t)z_{2}(t),\\
y_{i}(t)=z_{i}(t),i=1,2.\\
\end{array}
\right.
\]

We simulated system \eqref{e4} with different input sequences and initial state $x(0)\in\Delta_{8}$, and verified that for each index $i=1,2$, the output $y_i(t)$ varies with the input $u_i(t)$ when $u_j(t),j\neq i$ remains unchanged. For the sake of clarity and due to space limitation, we only show two sets of simulation results with initial states $x(0)=\delta_{8}^{5}$ in Figure \ref{Fig5.1}; one with $\{u_1(t)=\delta_{2}^{1}\}_{t=0}^{10}$ and a randomly selected input sequence $u_2$, the other one with $\{u_2(t)=\delta_{2}^{1}\}_{t=0}^{10}$, and a randomly selected input sequence $u_1$. The simulation results illustrate that the given BCN is one-step transition IO-decoupled.

\begin{figure}[h!]
     \centering
     \subfloat[]{\includegraphics[width=0.53\linewidth,height=0.5\linewidth,keepaspectratio]{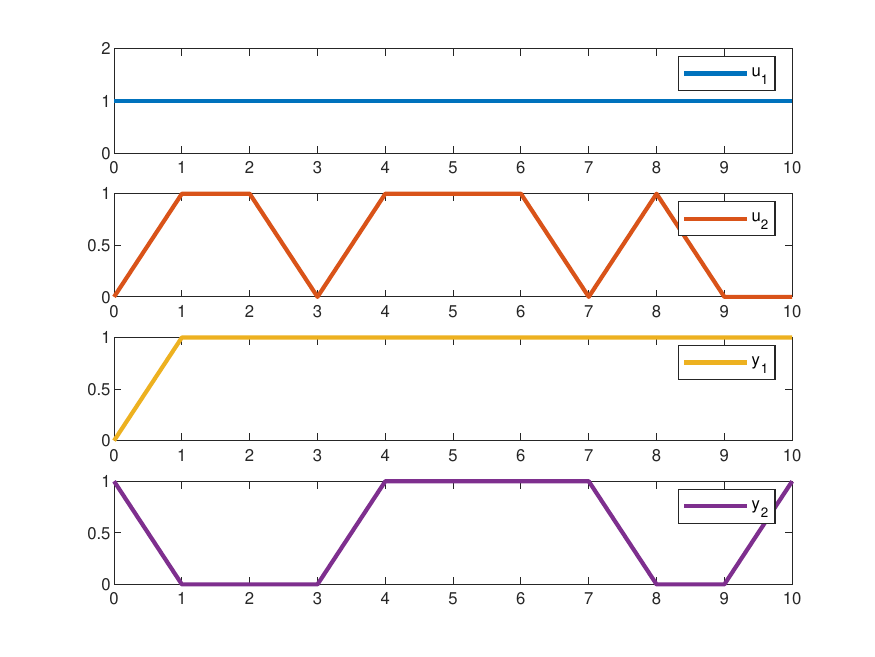}\label{Fig5.1a}}\hspace*{\fill}
    \subfloat[]{\includegraphics[width=0.53\linewidth,height=0.5\linewidth,keepaspectratio]{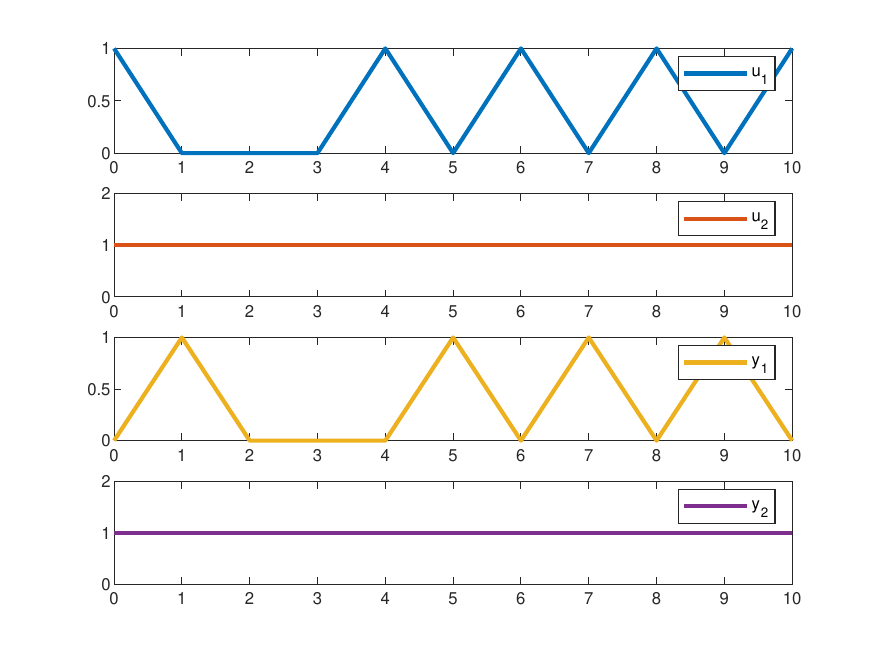}\label{Fig5.1b}}
     \caption{The evolution trajectories of $y_{1}(t)$ and $y_2(t)$ generated by a given initial state $x(0)=\delta_{8}^{5}$}
     \label{Fig5.1}
\end{figure}

Due to the fact that $(H_{1}\ast H_{2})\mathbf{1}_{8}=2\mathbf{1}_{4}$, an IO-decomposed form of BCN (\ref{e4}) exists.
Letting $z_{i}(t)=y_{i}(t)=H_{i}x(t),i=1,2$ and $z_{3}(t)=P_{3}x(t)$, one has $z(t)=Px(t)$, where $z(t)=\ltimes_{i=1}^{3}z_{i}(t)$ and $P=H_{1}\ast H_{2}\ast P_{3}$.
For $j\in\{1,\ldots,4\}$, defining $\Lambda_j$ as $\Lambda_j=\{x|(H_1\ast H_2)x=\delta_{4}^{j}\}$,
one obtains $\Lambda_{1}=\{\delta_{8}^{1},\delta_{8}^{3}\},\Lambda_{2}=\{\delta_{8}^{2},\delta_{8}^{4}\},\Lambda_{3}=\{\delta_{8}^{5},\delta_{8}^{7}\},\Lambda_{4}=\{\delta_{8}^{6},\delta_{8}^{8}\}$. For each state $\delta_{8}^{\lambda_{jk}}\in\Lambda_{j},j\in\{1,\ldots,4\},k=1,2$, take $\mathrm{Col}_{\lambda_{jk}}(P)=\delta_{4}^{j}\ltimes\delta_{2}^{k}$.
The nonsingular matrix $P$ is $P=\delta_{8}[1\ 3\ 2\ 4\ 5\ 7\ 6\ 8]$.
Then an IO-decomposed form of BCN (\ref{e4}) can be written as
\[
\left\{
\begin{array}{l}
z_{1}(t+1)=u_{1}(t),\\
z_{2}(t+1)=\delta_{2}[1\ 2\ 2\ 1]u_{2}(t)z_{2}(t),\\
z_{3}(t+1)=P_{3}L(I_{4}\otimes P^{\top})u(t)z(t),\\
y_{i}(t)=z_{i}(t),i=1,2,\\
\end{array}
\right.
\]
where $P_{3}=\delta_{2}[1\ 1\ 2\ 2\ 1\ 1\ 2\ 2]$ and $P_{3}L(I_{4}\otimes P^{\top})=\delta_{2}[2\ 1\ 2\ 1\ 2\ 2\ 2\ 2\ 1\ 2\ 1\ 2\ 2\ 2\ 2$
$2\ 1\ 1\ 1\ 1\ 2\ 1\ 2\ 1\ 1\ 1\ 1\ 1\ 2\ 1\ 2\ 1]$.
\end{example}

\begin{example}\label{ex3}\rm
Consider the following BCN with 2 input nodes, 4 state nodes and 2 output nodes:
\begin{align}\label{e5.1}
\scriptsize
\left\{
\begin{array}{l}
X_{1}(t+1)=(X_{1}(t)\vee(X_{2}(t)\leftrightarrow X_{4}(t)))\wedge U_{2}(t),\\
X_{2}(t+1)=(\neg U_{1}(t)\wedge (U_{2}(t)\wedge (X_{2}(t)\bar{\vee}X_{4}(t))\\
\ \ \ \ \ \ \ \ \ \ \vee(\neg U_{2}(t)\wedge(\neg X_{1}(t)\vee(X_{2}(t)\bar{\vee}X_{4}(t))))))\\
\ \ \ \ \ \ \ \ \ \ \vee(U_{1}(t)\wedge(X_{2}(t)\leftrightarrow X_{4}(t))),\\
X_{3}(t+1)=((X_{1}(t)\vee(X_{2}(t)\leftrightarrow X_{4}(t)))\wedge U_{2}(t))\\
\ \ \ \ \ \ \ \ \ \ \bar{\vee}((X_{1}(t)\vee U_{2}(t))\leftrightarrow U_{1}(t)),\\
X_{4}(t+1)=(U_{1}(t)\wedge U_{2}(t)\wedge X_{1}(t)\wedge(X_{2}(t)\leftrightarrow X_{4}(t)))\\
\ \ \ \ \ \ \ \ \ \ \vee(U_{1}(t)\wedge\neg U_{2}(t)\wedge(X_{2}(t)\bar{\vee}X_{4}(t)))\\
\ \ \ \ \ \ \ \ \ \ \vee(\neg U_{1}(t)\wedge U_{2}(t)\wedge(\neg X_{1}(t)\vee(X_{2}(t)\bar{\vee}X_{4}(t))))\\
\ \ \ \ \ \ \ \ \ \ \vee(\neg U_{1}(t)\wedge\neg U_{2}(t)\wedge(\neg X_{1}(t)\vee(X_{2}(t)\leftrightarrow X_{4}(t)))),\\
Y_{1}(t)=(X_{1}(t)\wedge X_{2}(t)\wedge\neg X_{3}(t)\wedge\neg X_{4}(t))\\
\ \ \ \ \ \ \ \ \ \ \vee(\neg X_{1}(t)\wedge\neg X_{2}(t))\vee(\neg X_{1}(t)\wedge\neg X_{4}(t)),\\
Y_{2}(t)=X_{1}(t)\bar{\vee}X_{3}(t),\\
\end{array}
\right.
\end{align}
where $U_{i},X_{j},Y_{k}\in\mathcal{D},i,k\in\{1,2\},j\in\{1,2,3,4\}$ and $\neg,\vee,\wedge,\leftrightarrow,\bar{\vee}$ are logical operators.
Let $u_{i}=\varphi(U_{i}),x_{j}=\varphi(X_{j})$ and $y_{k}=\varphi(Y_{k})$.
Then one has $u=u_{1}\ltimes u_{2},x=\ltimes_{i=1}^{4}x_{i}$ and $y=y_{1}\ltimes y_{2}$.
According to Lemma \ref{le1} and Lemma \ref{le2}, the algebraic form of BCN (\ref{e5.1}) is as follows:
\begin{align}\label{e5.2}
\left\{
\begin{array}{l}
x(t+1)=Lu(t)x(t),\\
y_{1}(t)=H_{1}x(t),\\
y_{2}(t)=H_{2}x(t),\\
\end{array}
\right.
\end{align}
where
\[
\small
\begin{array}{cccccccc}
L=\delta_{16}[3&8&3&8&8&3&8&3\\
\ \ \ \ \ \ \ \ \ \ \ 4&14&4&14&14&4&14&4\\
\ \ \ \ \ \ \ \ \ \ \ 10&13&10&13&13&10&13&10\\
\ \ \ \ \ \ \ \ \ \ \ 12&15&12&15&15&12&15&12\\
\ \ \ \ \ \ \ \ \ \ \ 6&1&6&1&1&6&1&6\\
\ \ \ \ \ \ \ \ \ \ \ 5&11&5&11&11&5&11&5\\
\ \ \ \ \ \ \ \ \ \ \ 15&12&15&12&12&15&12&15\\
\ \ \ \ \ \ \ \ \ \ \ 9&9&9&9&9&9&9&9],\\
H_{1}=\delta_{2}[2&2&2&1&2&2&2&2\\
\ \ \ \ \ \ \ \ \ \ \ 2&1&2&1&1&1&1&1],\\
H_{2}=\delta_{2}[2&2&1&1&2&2&1&1\\
\ \ \ \ \ \ \ \ \ \ \ 1&1&2&2&1&1&2&2].\\
\end{array}
\]
It is obtained from $H_{1}$ and $H_{2}$ that $\Gamma_{11}=\{\delta_{16}^{4},\delta_{16}^{10},\delta_{16}^{12},$
$\delta_{16}^{13},\delta_{16}^{14},\delta_{16}^{15},\delta_{16}^{16}\},\Gamma_{12}=\{\delta_{16}^{1},\delta_{16}^{2},\delta_{16}^{3},\delta_{16}^{5},\delta_{16}^{6},\delta_{16}^{7},\delta_{16}^{8},\delta_{16}^{9},$
$\delta_{16}^{11}\},\Gamma_{21}=\{\delta_{16}^{3},\delta_{16}^{4},\delta_{16}^{7},\delta_{16}^{8},\delta_{16}^{9},\delta_{16}^{10},\delta_{16}^{13},\delta_{16}^{14}\},\Gamma_{22}=\{\delta_{16}^{1},$
$\delta_{16}^{2},\delta_{16}^{5},\delta_{16}^{6},\delta_{16}^{11},\delta_{16}^{12},\delta_{16}^{15},\delta_{16}^{16}\}$. Since  \[\mathrm{sgn}(\sum\limits_{\delta_{16}^{\alpha}\in\Gamma_{ij}}H_{i}LV(\Omega_{ik})\delta_{16}^{\alpha})=\mathbf{1}_{2}\] holds for $i,j,k\in\{1,2\}$, one has
\[
\Xi^{i}=\left[
\begin{array}{cc}
\mathbf{1}_{2}&\mathbf{1}_{2}\\
\mathbf{1}_{2}&\mathbf{1}_{2}\\
\end{array}
\right]_{,}
\]
which results in $(\delta_{2}^{l})\Xi^{i}\mathbf{1}_{2}=2\mathbf{1}_{2},i,l\in\{1,2\}$.
BCN (\ref{e5.2}) is not one-step transition IO-decoupled as per Theorem \ref{th1}.

Figure \ref{Fig5.2} shows two sets of simulation results with initial state $x(0)=\delta_{16}^{6}$; one with $\{u_1(t)=\delta_{2}^{1}\}_{t=0}^{10}$ and a randomly selected input sequence $u_2$, the other one with $\{u_2(t)=\delta_{2}^{1}\}_{t=0}^{10}$ and a randomly selected input sequence $u_1$. One can clearly see that the output $y_2$ generated by $x(0)=\delta_{16}^{6}$ varies with both inputs $u_1$ and $u_2$, which confirms that BCN (\ref{e5.2}) is not one-step transition IO-decoupled.

\begin{figure}[h!]
     \centering
    \subfloat[]{\includegraphics[width=0.53\linewidth,height=0.5\linewidth,keepaspectratio]{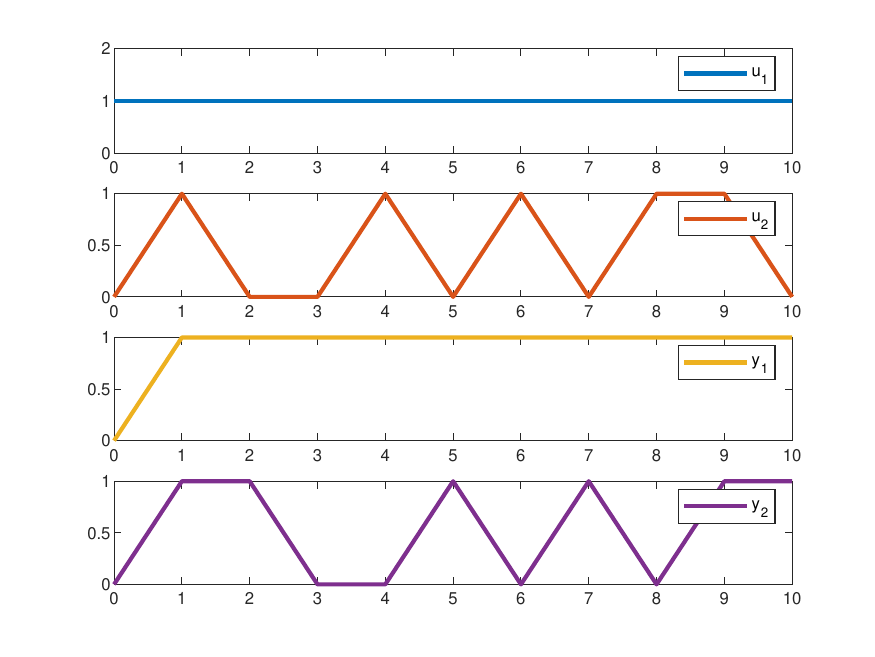}\label{Fig5.2a}}\hspace*{\fill}
    \subfloat[]{\includegraphics[width=0.53\linewidth,height=0.5\linewidth,keepaspectratio]{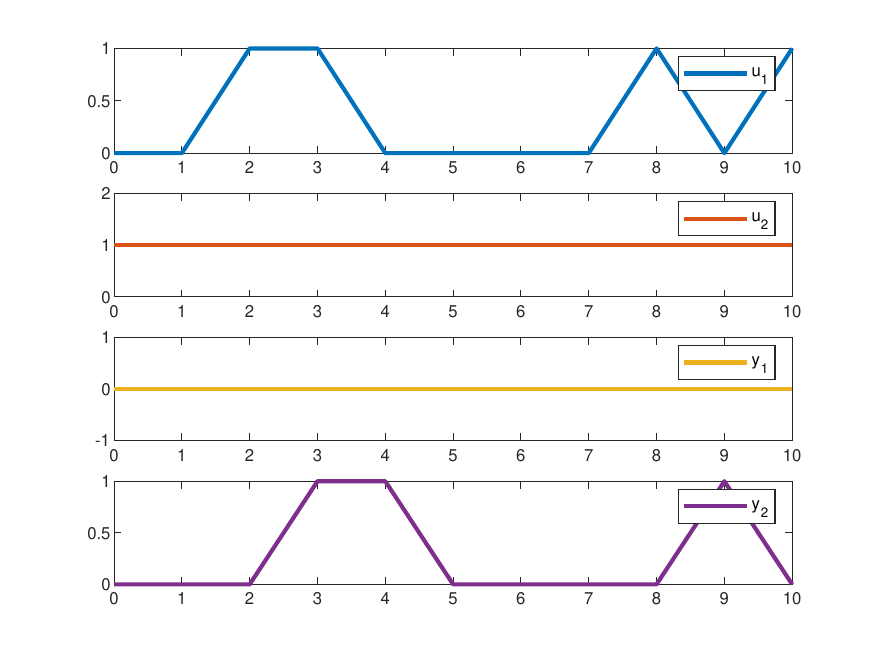}\label{Fig5.2b}}
     \caption{The evolution trajectories of $y_{1}(t)$ and $y_2(t)$ generated by a given initial state $x(0)=\delta_{16}^{6}$}
\label{Fig5.2}
\end{figure}

In the following, we use the proposed approach to design a state feedback control to make BCN (\ref{e5.2}) one-step transition IO-decoupled.
Auxiliary matrices $\hat{K}_{1}$ and $\hat{K}_{2}$ are computed according to (\ref{e3.2.4}), \ie,
\[
\begin{array}{l}
\hat{K}_{1}\delta_{16}^{i}=\delta_{2}[2\ 1\ 2\ 1],i\in\{1,\ldots,8\},\\
\hat{K}_{1}\delta_{16}^{i}=\delta_{2}[1\ 1\ 2\ 2],i\in\{9,\ldots,16\},\\
\hat{K}_{2}\delta_{16}^{i}=\delta_{2}[1\ 1\ 2\ 2],i\in\{1,\ldots,8\},\\
\hat{K}_{2}\delta_{16}^{i}=\delta_{2}[1\ 2\ 2\ 1],i\in\{9,\ldots,16\}.\\
\end{array}
\]
It is easy to notice that $\mathrm{sgn}(\hat{K}_{i}\delta_{16}^{\alpha}\mathbf{1}_{4})=\mathbf{1}_{2}$ holds for $i\in\{1,2\}$ and $\alpha\in\{1,\ldots,16\}$.
Thus, we obtain $\hat{K}$, where
\[
\begin{array}{l}
\hat{K}\delta_{16}^{i}=\delta_{4}[2\ 4\ 1\ 3],i\in\{1,\ldots,8\},\\
\hat{K}\delta_{16}^{i}=\delta_{4}[1\ 2\ 4\ 3],i\in\{9,\ldots,16\}.
\end{array}
\]
Due to the fact that $\hat{K}\in\mathcal{L}_{4\times 16}$, the state feedback matrix is designed as $K=\hat{K}$.
Substituting $u(t)=Kx(t)v(t)$ into BCN (\ref{e5.2}), the resulting closed-loop system is obtained
\begin{equation}\label{e5.3}
\left\{
\begin{array}{l}
x(t+1)=\hat{L}v(t)x(t),\\
y_{1}(t)=H_{1}x(t),\\
y_{2}(t)=H_{2}x(t),\\
\end{array}
\right.
\end{equation}
where $\hat{L}$ is computed via $\hat{L}=LKW_{[2^{n},2^{n+m}]}\Phi$ and $\Phi=\mathrm{diag}\{\delta_{16}^{1},\ldots,\delta_{16}^{16}\}$, \ie,
\[
\small
\begin{array}{cccccccc}
\hat{L}=\delta_{16}[10&13&10&13&13&10&13&10\\
\ \ \ \ \ \ \ \ \ \ 4&14&4&14&14&4&14&4\\
\ \ \ \ \ \ \ \ \ \ 15&12&15&12&12&15&12&15\\
\ \ \ \ \ \ \ \ \ \ 12&15&12&15&15&12&15&12\\
\ \ \ \ \ \ \ \ \ \ 3&8&3&8&8&3&8&3\\
\ \ \ \ \ \ \ \ \ \ 9&9&9&9&9&9&9&9\\
\ \ \ \ \ \ \ \ \ \ 6&1&6&1&1&6&1&6\\
\ \ \ \ \ \ \ \ \ \ 5&11&5&11&11&5&11&5].\\
\end{array}
\]
Then the Boolean matrices $\widehat{\Xi}^{1}$ and $\widehat{\Xi}^{2}$ given by
\[
\widehat{\Xi}^{i}=\left[
\begin{array}{cc}
\delta_{2}^{1}&\delta_{2}^{2}\\
\delta_{2}^{1}&\delta_{2}^{2}\\
\end{array}
\right]_{,}
\]
where $i\in\{1,2\}$.
The resulting closed-loop system (\ref{e5.3}) is one-step transition IO-decoupled as per Theorem \ref{th1}.
The mapping between $y_{i}(t+1)$ and $v_{i}(t),y_{i}(t)$ can be depicted via
\[
y_{i}(t+1)=\delta_{2}[1\ 2\ 1\ 2]y_{i}(t)v_{i}(t),
\]
where $i=1,2$.
Let $z_i(t)=y_i(t),i=1,2$. Then a one-step transition IO-decoupled canonical form of the resulting closed-loop system \eqref{e5.3} is as follows:
\begin{equation}
\left\{
\begin{array}{l}
z_{i}(t+1)=v_{i}(t),\\
y_{i}(t)=z_{i}(t),\\
\end{array}
\right.
\end{equation}
where $i=1,2$.
Moreover, due to the fact that
\[(H_{1}\ast H_{2})\mathbf{1}_{16}=[4\ 3\ 4\ 5]^{\top}\neq 4\mathbf{1}_{4},\]
the resulting closed-loop system (\ref{e5.3}) does not have an IO-decomposed form as per Theorem \ref{th5}.

For each initial state $x(0)\in\Delta_{16}$, we tested different input sequences and found that for each index $i=1,2$, the output $y_i(t)$ varies with the input $v_i(t)$ when $v_j(t),j\neq i$ remains unchanged. Figure \ref{Fig5.3a} shows the simulation results obtained with $x(0)=\delta_{16}^{6}$, $\{v_1(t)=\delta_{2}^{1}\}_{t=0}^{10}$ and a randomly selected input sequence $v_2$.  Figure \ref{Fig5.3b} shows the simulation results obtained with $x(0)=\delta_{16}^{6}$, $\{v_2(t)=\delta_{2}^{1}\}_{t=0}^{10}$, and a randomly selected input sequence $v_1$. The results clearly show that resulting closed-loop system (\ref{e5.3}) is one-step transition IO-decoupled.

\begin{figure}[h!]
    \centering
     \subfloat[]{\includegraphics[width=0.53\linewidth,height=0.5\linewidth,keepaspectratio]{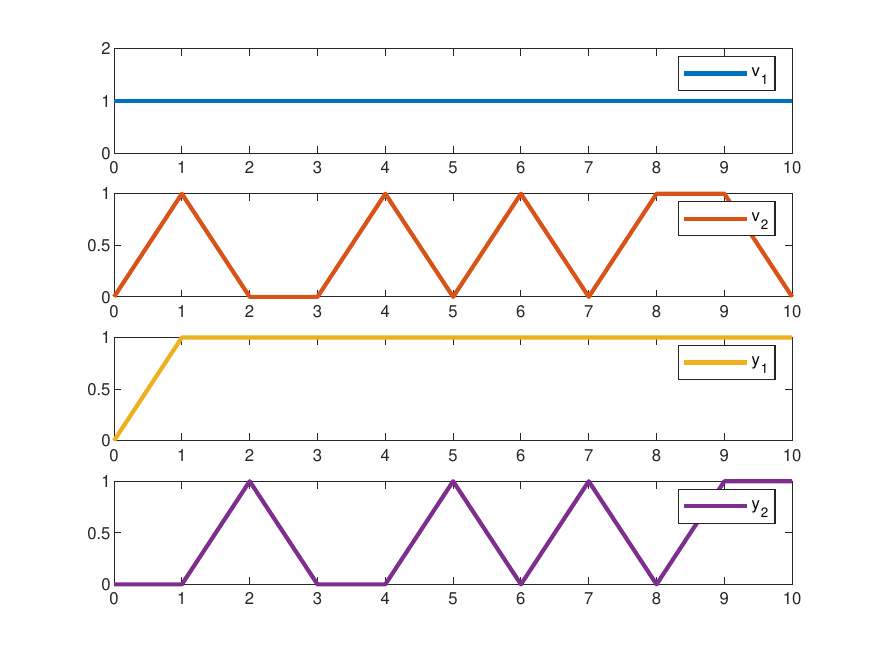}\label{Fig5.3a}}\hspace*{\fill}
     \subfloat[]{\includegraphics[width=0.53\linewidth,height=0.5\linewidth,keepaspectratio]{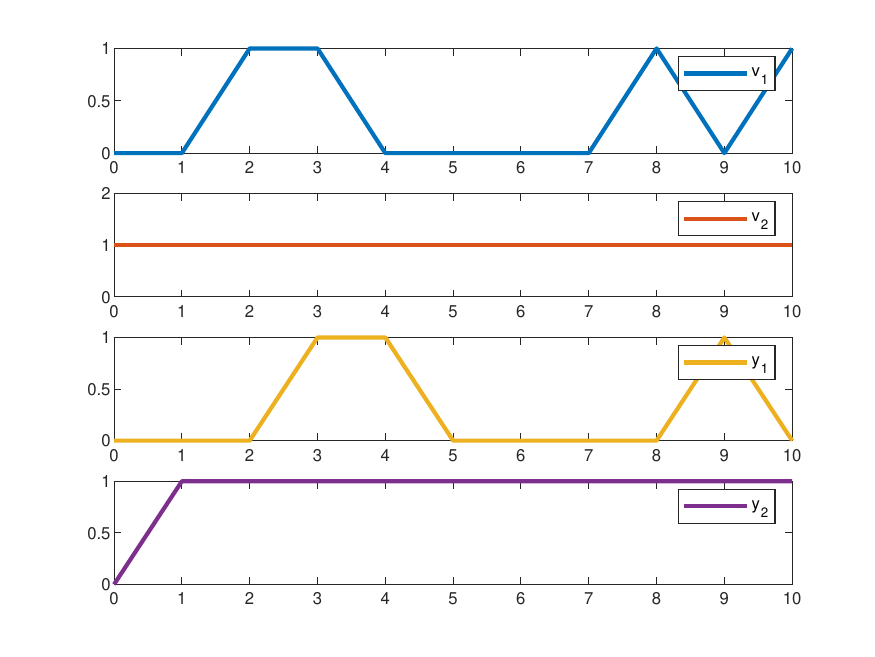}\label{Fig5.3b}}
    \caption{The evolution trajectories of $y_{1}(t)$ and $y_2(t)$ generated by a given initial state $x(0)=\delta_{16}^{6}$}
\label{Fig5.3}
\end{figure}
\end{example}

\section{Conclusion}\label{sec6}
Some necessary and sufficient conditions for one-step transition IO-decoupling of BCNs via state feedback have been derived.
As an instrumental tool for our IO-decoupling state feedback control design, an explicit mapping between outputs and inputs, referred to as the one-step transition decoupled canonical form, has been established. In contrast to the existing approaches in the literature, such as those in \cite{Fushihua2017,Panjinfeng2018}, our approach is viable even when the BCN does not admit an IO-decomposed form. It is also shown that our approach allows to obtain, in a straightforward manner the IO-decomposed form if it exists. Our proposed approach has also been put in perspective, with respect to the existing IO-decoupling definitions in the literature, and thorough discussion is provided including some important new and rigorously proven results.

Although the proposed approach focuses on BCNs with bi-valued nodes (\ie, each node takes a value of either zero or one), the obtained results could be generalized to multi-valued logical control networks (\ie, each node has $k,k>2$ values).
The IO-decoupling problem of multi-valued stochastic logical control networks is an interesting problem that would be part of our future investigations.

It is worth poiting out that the proposed approach has an exponential computational complexity caused by the algebraic form of BCNs involving $2^{n}$ dimensional state $x(t)$ associated to the $n$ dimensional state $X(t)$. Note also that many problems involving BCNs, such as stability and observability, are NP-hard. Consequently, at this point in time and with the available tools, the high computational complexity seems to be intrinsic to BCNs. In our future investigations, we will explore the possibility of developing new tools and techniques to reduce the aforementioned computational complexity.


\end{document}